\documentclass[11pt]{article}

\usepackage{latexsym,amssymb}
\RequirePackage{graphicx}
\usepackage{subcaption}
\usepackage{xcolor}
\usepackage{amsmath}
\usepackage{amsthm}
\usepackage{amsfonts}
\usepackage{ dsfont }
\usepackage{hyperref}
\usepackage{cleveref}
\usepackage[a4paper,margin=1in]{geometry}

\allowdisplaybreaks

\newtheorem{theorem}{Theorem}
\newtheorem{lemma}[theorem]{Lemma}

\newtheorem{proposition}[theorem]{Proposition}

\newtheorem{remark}[theorem]{Remark}

\newtheorem{assumption}{Assumption}

\title{Convergence and non-asymptotic error analysis for kinetic Langevin samplers using the exact harmonic Langevin integrator}

\author{Katharina Schuh\thanks{Institute of Analysis and Scientific Computing, TU Wien, Wiedner Hauptstra\ss e 8--10, 1040 Wien, Austria 
  ({katharina.schuh@tuwien.ac.at}
  ).}
}

\usepackage{amsopn}

\def\rmd{\mathrm{d}}

\def\1{\mathds{1}}

\ifpdf
\hypersetup{
  pdftitle={Convergence of kinetic Langevin samplers},
  pdfauthor={K. Schuh}
}
\fi

\begin{document}

\maketitle

\begin{abstract}
We propose a novel kinetic Langevin sampler based on a specific splitting scheme using the exact harmonic Langevin integrator. For strongly log-concave target measures, the sampler exploits a decomposition of the strongly convex potential into a quadratic part and a convex perturbation with Lipschitz continuous gradient.
For the resulting first- and second-order schemes associated with this splitting we establish convergence rates in $L^2$-Wasserstein distance as well as non-asymptotic error bounds.
In particular, the contraction rate is of the same order as that of the underlying continuous dynamics.
To achieve $\varepsilon$-accuracy, the required step size for the second-order scheme is comparable to that of established splitting schemes such as OBABO or UBU, which are widely used in machine learning and molecular dynamics.



\end{abstract}

\section{Introduction}

In this work we consider a novel splitting scheme derived from the (kinetic) Langevin dynamics $(X_t,V_t)_{t\ge 0}$ given by the second order stochastic differential equation
\begin{align}\label{eq:LD}
\begin{cases}
\rmd X_t= V_t \rmd t \\
\rmd V_t= (-\nabla U(X_t) -\gamma V_t) \rmd t +\sqrt{2\gamma}\rmd B_t,
\end{cases}
\end{align}
where $\gamma>0$ denotes the friction parameter, $(B_t)_{t\ge 0}$ the $d$-dimensional standard Brownian motion and $U:\mathbb{R}^d\to\mathbb{R}$ the potential function. 

The dynamics describes in statistical physics a particle with position and momentum that undergoes damping and random collisions. A central property is that the Boltzmann-Gibbs measure
\begin{align*}
\mu(\rmd x\rmd v)\propto \exp(-U(x)-|v|^2/2)\rmd x\rmd v
\end{align*}
forms its unique invariant probability measure, to which the dynamics converges under suitable conditions on the potential $U$ (see e.g. \cite[Proposition 6.1]{Pa2014}). This observation motivates the use of the continuous Langevin dynamics and numerical implementable approximations, for sampling from the target measure $\mu$ or its position marginal $\mu^x\propto \exp(-U(x))\rmd x$.
This work contributes towards this approach by implementing a novel splitting scheme exploiting the special structure of the potential together with the exact harmonic Langevin integrator, and by providing its long-time analysis in the setting where the target measure $\mu$ is a high-dimensional log-concave measure.
In particular, we consider a decomposition of the strongly convex function $U$ into a quadratic term and a convex remainder term, which is understood below as a perturbation term, i.e.,
\begin{align} \label{eq:U}
U(x)=-\frac{1}{2}x^T K x+ G(x).
\end{align}

\paragraph{Overview of the literature}
Sampling efficiently from high-dimensional log-concave distributions is a central problem of significant interests in many applications including statistical physics \cite{LeSt2016}, machine learning \cite{AnDeDoJo2003} or Bayesian sampling \cite{GeCaStRu1995}. 

The Langevin dynamics \eqref{eq:LD} and its first-order variant, which is known as the overdamped Langevin dynamics and is used to sample $\mu^x$, form a basis for many sampling algorithms.

For the overdamped Langevin dynamics convergence and contraction rates are shown both by functional inequalities \cite{BaGeLe2014,CaGu2009} and probabilistic tools such as coupling techniques \cite{Eb2016}. 
Considering numerical approximation, the unadjusted Langevin algorithm (ULA) and its variants form a very popular class of Langevin-based sampling methods. Their convergence and non-asymptotic complexity bounds have been extensively studied (see e.g. \cite{Da2017,DuMo2017,DuMo2019}). However for a $\kappa$-strongly log-concave distribution a contraction rate of order $\mathcal{O}(\kappa)$ is the optimal achievable rate for first order methods which motivates to consider sampling algorithms based on higher-order dynamics, in particular second-order Langevin dynamics.

The continuous second-order Langevin dynamics given in \eqref{eq:LD} was object of many works (see e.g. \cite{BaCaGu2008, CaLuWa2023, ChChBaJo2018,DaRi2020, DoMoSc2009, DoMoSc2015,EbGuZi2019b,Sc2024, Vi2009}) and for strongly $\kappa$-convex potentials contraction in $L^2$-Wasserstein distance with rates of order $\mathcal{O}(\sqrt{\kappa})$ can be achieved. 
The same rate carries over to a broad class of kinetic Langevin sampler building on numerical approximations of the continuous dynamics.
Besides the classical Euler-Maruyama discretization, splitting schemes form a popular family of kinetic Langevin sampler. In \cite{ChChBaJo2018,DaRi2020,LePaWh2024a,SaZy2021} the long-time behaviour of various kinetic Langevin sampler is analysed and contraction rates matching to the one of the continuous dynamics are obtained. 

The contraction results highlight the potential acceleration obtained by incorporating momentum variables into the sampling procedure.
In parallel, the Hamiltonian Monte Carlo algorithm \cite{Ne1995} forms another prominent momentum-based algorithm by using the Hamiltonian dynamics. For the convergence analysis of the exact Hamiltonian Monte Carlo method we refer to \cite{MaSm2021,BoEbZi2020}. In \cite{GoBrMaMo2025} the generalized HMC, a unified framework for both HMC and the splitting scheme OBABO is introduced and its long-time behaviour is analysed.

The numerical discretization of both the Hamiltonian based and the Langevin dynamics based algorithms introduce an additional error. Bounds on the efficiency of the resulting algorithms are studied for the Langevin dynamics in \cite{CaDuMoSt2023,ChChBaJo2018,LePaWh2024a,Mo2021} and for HMC in \cite{BoSc2023,ChVe2022,MaSm2021}. 
Alternatively, the error can be avoided by an additional Metropolis acceptance-rejection step which allows to sample exactly from the correct target measure \cite{MeRoRoTeTe1953,Ha1970}. 
The Metropolis-adjusted algorithms are not scope of this work and we focus on bounding the additional bias term induced by inexact sampling and studying the dependence on the dimension and condition number of the obtained complexity guarantees. 

Recent advances have been made in the analysis of sampling non-log-concave distributions. Often sampling guarantees are obtained under relaxed dissipativity guarantees by using reflection couplings \cite{Eb2016,EbGuZi2019b}.
We assume that the proposed kinetic Langevin sampler is applicable for sampling non log-concave distributions and sampling guarantees hold under a similar setup as in \cite{ScWh2024}.

Finally, we remark that in many practical applications the exact gradient evaluation is not available or computationally too costly. In these settings,  one typically makes use of a stochastic gradient approximation. In the context of Bayesian Learning, \cite{WeTe2011} proposed sampling based on mini-batches.
The convergence behaviour and complexity bounds for kinetic Langevin samplers with stochastic gradients are analysed in \cite{GoBrMaMo2025,LePaWh2024}. 



\paragraph{Our main contribution} The contribution of the paper is three-fold. 
\begin{itemize}
\item Firstly, we propose a new splitting of \eqref{eq:LD} into two components using the decomposition of the potential \eqref{eq:U} and present the precise numerical implementation of the corresponding first-order and second-order splitting scheme. 
To implement the component given by the Langevin dynamics with a linear drift, we distinguish between the overdamped, the underdamped and the critical cases (see Proposition~\ref{prop:numscheme}). 
\item 
Secondly, we study the long-time behaviour for the proposed schemes.
Under the assumption that the potential $U$ is $\kappa$-strongly convex and has Lipschitz continuous gradients, we establish epxonential contraction.
More precisely, for any initial measures $\nu,\eta\in\mathcal{P}_2(\mathbb{R}^{2d})$ and for the transition kernel $\pi_h$ associated to the numerical schemes with discretization parameter $h>0$, we prove
\begin{align*}
\mathcal{W}_2(\nu  \pi_h^k, \eta \pi_h^k)\le M_1 e^{-c h k } \mathcal{W}_2(\nu,\eta),
\end{align*}
where the contraction rate $c$ is given by $c=\min(\gamma/8, \kappa \gamma^{-1}/4)$ (see Theorem~\ref{thm:conv} and Theorem~\ref{thm:conv_PGP2}). The constant $M_1>0$ which is independent of $h$, $k$ and $d$ appears since contraction is first shown in a $L^2$-Wasserstein distance with respect to a suitably twisted distance that is equivalent to the Euclidean distance. The result holds for friction parameters satisfying $\gamma\ge \sqrt{2L_G}$ where $L_G$ denotes the Lipschitz constant of the gradient of the perturbation term $G$. The condition is consistent with known contraction results for the continuous Langevin dynamics \cite{Sc2024} and with the observation that in the Gaussian case the restriction on $\gamma$ disappears.

\item Thirdly, since these unadjusted schemes do not directly sample the target measure, we study the size of the induced error. In particular, we bound the distance between the true target measure and the invariant probability measure of the schemes by an error of order $h$ for the first-order scheme and of order $h^2$ for the symmetric second-order scheme. 
Combining the error and the contraction result we obtain guarantees on the number of steps of the kinetic sampler and gradient evaluations to approach the true target measure $\varepsilon$-close in $L^2$ Wasserstein distance (see Theorem~\ref{thm:complex} and Theorem~\ref{thm:complex_PGP}).

\end{itemize}

We note that the contraction results and error bounds obtained here are comparable to those of the standard splitting schemes as the OBABO or UBU scheme (\cite{LePaWh2024a, SaZy2021}). In particular the proposed splitting scheme is beneficial for Bayesian sampling when the prior is given by a Gaussian measure, since the structure of the log-likelihood function matches with the remainder term of the splitting. 

\paragraph{Outlook of the paper}
In Section~\ref{sec:ass_splitting} we present the new splitting scheme. In Section~\ref{sec:conv_res} we give the results on the convergence behaviour, before we provide the corresponding strong accuracy bounds and the complexity analysis in Section~\ref{sec:compl_res}. 
The proofs are postponed to Section~\ref{sec:proofs}.

\section{Framework and splitting schemes for the kinetic Langevin sampler} \label{sec:ass_splitting}

\subsection{Assumptions on the potential}

As the construction of the kinetic Lange\-vin sampler uses the structure of the potential function $U$, we first state the assumptions for $U$:

\begin{assumption}\label{ass:U}
There exists a symmetric positive definite matrix $K$ and a convex, differentiable function $G$ which is gradient Lipschitz such that
\begin{align*}
U(x)=\frac{1}{2}x^T K x + G(x).
\end{align*}
We denote the smallest eigenvalue of $K$ by $\kappa>0$, the Lipschitz constant of the linear function $x\mapsto Kx$ by $L_K$ and the Lipschitz constant of $\nabla G$ by $L_G$.

\end{assumption}
Assumption~\ref{ass:U} is satisfied for $\kappa$-strongly convex potentials with Lipschitz continuous gradients. A natural decomposition of the potential $U$ is given by $U(x)=\frac{1}{2}x^T K x + G(x)$ with $K=\kappa I_d$ and convex function $G(x)=U(x)-\frac{\kappa}{2}|x|^2$. But, the decomposition is not unique. In particular, one aims to find a splitting such that the Lipschitz constant $L_G$ of the convex remainder term $G$ is minimised.

We note that $\nabla U$ is Lipschitz continuous with $\|\nabla U\|_{\mathrm{Lip}}\le L_K+L_G$.

\subsection{Splitting schemes}

Recently, advanced numerical schemes of the Langevin dynamics are obtained by using splitting methods. The key idea of these methods is based on splitting the dynamics into components which can be integrated exactly, see \cite{McQu2002} for an overview on splitting methods. Due to their good performance they are frequently used in molecular dynamics \cite{BuPa2007}. 

A common splitting is given by considering the Ornstein-Uhlenbeck ($\mathcal{O}$) and the position part ($\mathcal{A}$) and the velocity part ($\mathcal{B}$) corresponding to the Hamiltonian part of \eqref{eq:LD}. Then, integrating exactly this terms successively for some time $h>0$ leads to first-order schemes of the form $\mathcal{OBA}$, $\mathcal{BAO}$ ... 
Building symmetric higher-order splitting schemes like $\mathcal{OBABO}$ and $\mathcal{BAOAB}$ provide numerical splitting schemes with errors of better order (see e.g. \cite{LeMa2013,LeMaSt2016, LePaWh2024a}, \cite{Mo2021}). 
We also refer to \cite{ChMo2025} for the asymptotic analysis in the non-strongly convex setting.
Recently the splitting where the parts $\mathcal{A}$ and $\mathcal{O}$ are combined to the part $\mathcal{U}$ became famous resulting in the splitting schemes $\mathcal{UB}$ and $\mathcal{UBU}$ \cite{Za2017}. The symmetric second-order method has weak and strong error of order two \cite{SaZy2021, PaWh2024}.

Let $h>0$ be the discretization size. In this work we consider the following two splitting components:
\begin{align*}
\begin{pmatrix}
\rmd x \\ \rmd v
\end{pmatrix}
=
\underbrace{\begin{pmatrix}
0 \\ - \nabla G(x) \rmd t
\end{pmatrix}}_{\mathcal{P}}
+
\underbrace{
\begin{pmatrix}
v \rmd t 
\\
-Kx \rmd t -\gamma \rmd t + \sqrt{2\gamma} \rmd B_t
\end{pmatrix}}_{\mathcal{G}}.
\end{align*}
This splitting is also considered in numerical simulations of the Jansen and Rit neural mass model \cite{AbBuHi2017}.
The exact integration of the step $\mathcal{P}$, where only the perturbation $\nabla G$ is considered, is given by
\begin{align}\label{eq:pstep}
\mathcal{P}(x,v,h)=(x,v-h\nabla G(x)).
\end{align} 
For the exact integration of the part $\mathcal{G}$, we observe that it corresponds to the exact solution of the continuous Langevin dynamics $(\tilde{X}_t,\tilde{V}_t)_{t\ge 0}$ with a quadratic potential given by
\begin{align*}
\begin{pmatrix}
\rmd \tilde{X}_t 
\\ \rmd \tilde{V}_t
\end{pmatrix}=A \begin{pmatrix}
 \tilde{X}_t 
\\ \tilde{V}_t
\end{pmatrix}\rmd t + \sqrt{2\gamma} \rmd \begin{pmatrix}
0 \\ B_t
\end{pmatrix}, \qquad \text{with } 
A=\begin{pmatrix}
0_d & 1_d \\ -K & -\gamma 1_d
\end{pmatrix},
\end{align*}
where $(B_t)_{t\ge 0}$ is a $d$-dimensional standard Brownian motion.
The explicit solution $(\tilde{X}_h,\tilde{V}_h)$ of the exact harmonic Langevin integrator at time $h>0$ is given by
\begin{align} \label{eq:exactsol_d}
\begin{pmatrix}
\tilde{X}_h \\ \tilde{V}_h
\end{pmatrix}
= e^{Ah} 
\begin{pmatrix}
\tilde{X}_0 \\ \tilde{V}_0
\end{pmatrix}+ \int_0^h e^{A(h-s)} \sqrt{2\gamma} \rmd \begin{pmatrix}
 0
\\  B_s
\end{pmatrix}.
\end{align}
To provide the exact numerical scheme for $(\tilde{X}_h,\tilde{Y}_h)$, we impose the following assumption:
\begin{assumption}\label{ass:diag}
The matrix $K$ is of diagonal form, i.e., there exists $k_1,..., k_d$ with $k_j\ge \kappa$ such that $K=\mathrm{diag}(k_1,...,k_d)$.
\end{assumption}

In this case, the solution can be considered componentwise and for each $j\in\{1,\ldots,d\}$ the matrix exponential $\mathcal{E}_j(t)=e^{A_jt}\in\mathbb{R}^{2\times 2}$ with $A_j=\begin{pmatrix}
0 & 1 \\ -k_j & -\gamma
\end{pmatrix}$ can be computed separately.
Then, the matrix exponential $e^{At}\in\mathbb{R}^{2d\times 2d}$ is given by
\begin{align*}
e^{At}= \begin{pmatrix}
D_{1,1}(t) & D_{1,2}(t) \\ D_{2,1}(t) & D_{2,2}(t)
\end{pmatrix}
\text{ with } D_{k,l}(t)=\mathrm{diag}((\mathcal{E}_j(t))_{k,l}: j=1,\ldots,d), \ k,l\in\{1,2\}.
\end{align*}

The following numerically implementable discretization represents the exact integration of $\mathcal{G}$. We note that this representation is not unique.

\begin{proposition}[Numerical representation of step $\mathcal{G}$] \label{prop:numscheme}
Let $h>0$. Let $\xi, \zeta \sim \mathcal{N}(0_d,I_d)$ be two independent $d$-dimensional standard normally distributed random variables. Suppose Assumption~\ref{ass:diag} holds true. Let $(x,v)\in\mathbb{R}^{2d}$. Then the exact integration of step $\mathcal{G}$ is given componentwise  by 
\begin{align} \label{eq:stepG}
\mathcal{G}_j((x,v),h,\xi,\zeta)= \mathbf{A}_j(h)\begin{pmatrix}
x^j \\ v^j 
\end{pmatrix}+ \mathbf{B}_j(h) \begin{pmatrix}
\xi^j \\ \zeta^j
\end{pmatrix}\in \mathbb{R}^2, \qquad  j=1, \ldots, d,
\end{align}
where $\mathbf{A}_j(h)$ and $\mathbf{B}_j(h)$ are of the following form: \\
In the \textbf{overdamped case} ($\gamma^2 >4k_j$):
\begin{align*}
& \mathbf{A}_j(h)=e^{-\frac{\gamma}{2}h} \begin{pmatrix}
\frac{1}{\omega}\sinh(\frac{\gamma}{2}\omega h)+\cosh(\frac{\gamma}{2}\omega h) & \frac{2}{\gamma\omega}\sinh(\frac{\gamma}{2}\omega h)
\\ -k_j\frac{2}{\gamma\omega} \sinh(\frac{\gamma}{2}\omega h) & -\frac{1}{\omega}\sinh(\frac{\gamma}{2}\omega h)+\cosh(\frac{\gamma}{2}\omega h)
\end{pmatrix} ,
\\ & \mathbf{B}_j(h)=
\begin{pmatrix}
\frac{\sqrt{2}}{\gamma \omega}
\frac{(1+\omega)(1-e^{-\gamma h} )-(1-e^{-\gamma(1+\omega)h})}{\sqrt{(1-e^{-\gamma(1+\omega) h})(1+\omega) } }
& \frac{\sqrt{2}}{\gamma \omega}\sqrt{\frac{1-e^{-\gamma(1-\omega)h}}{1-\omega}-\frac{(1-e^{-\gamma h})^2(1+\omega)}{1-e^{-\gamma(1+\omega)h})}}
\\  \frac{1+\omega}{\sqrt{2}\omega}\frac{(1-e^{-\gamma(1+\omega)h})-(1-\omega)(1-e^{-\gamma h})}{\sqrt{(1-e^{-\gamma(1+\omega)h})(1+\omega)}}
&  \frac{\omega-1}{\sqrt{2}\omega} \sqrt{\frac{1-e^{-\gamma(1-\omega)h}}{1-\omega}-\frac{(1-e^{-\gamma h})^2(1+\omega)}{1-e^{-\gamma(1+\omega)h}}}
\end{pmatrix}
\end{align*}
with $\omega=\sqrt{|1-4k_j/\gamma^2|}$. \\
In the \textbf{underdamped case} ($\gamma^2 <4k_j$):
\begin{align*}
& \mathbf{A}_j(h)=e^{-\frac{\gamma}{2}h} \begin{pmatrix}
\frac{1}{\omega}\sin(\frac{\gamma}{2}\omega h)+\cos(\frac{\gamma}{2}\omega h) & \frac{2}{\gamma\omega}\sin(\frac{\gamma}{2}\omega h)
\\ -k_j\frac{2}{\gamma\omega} \sin(\frac{\gamma}{2}\omega h) & -\frac{1}{\omega}\sin(\frac{\gamma}{2}\omega h)+\cos(\frac{\gamma}{2}\omega h)
\end{pmatrix}, 
\\ & \mathbf{B}_j(h)=\begin{pmatrix}
\frac{2}{\gamma \omega}\sqrt{\frac{b }{1+\omega^2} } & 0 \\  \frac{e^{-\gamma h} (1-\cos(\gamma \omega h)) \sqrt{1+\omega^2}  }{\sqrt{b}}  & \frac{\sqrt{  (1-e^{-\gamma h})^2\omega^2  - 2 e^{-\gamma h} (1-\cos( \gamma \omega h ))}}{\sqrt{b}} 
\end{pmatrix}
\end{align*}
with
$b=(1+\omega^2)(1-e^{-\gamma h})-1+e^{-\gamma h}\cos(\omega \gamma h)-\omega e^{-\gamma h}\sin(\omega\gamma h)$
%
and \\ $\omega=\sqrt{|1-4k_j/\gamma^2|}$.\\
In the \textbf{critical case} ($\gamma^2 =4k_j$):
\begin{align*}
& \mathbf{A}_j(h)=e^{-\frac{\gamma}{2} h} \begin{pmatrix}
\frac{\gamma}{2}h+1& h
\\ -k_j h & -\frac{\gamma}{2}h +1
\end{pmatrix} \qquad \text{and} 
\\ & 
\mathbf{B}_j(h)=\begin{pmatrix}
\frac{\sqrt{4(1-e^{-\gamma h})-2e^{-\gamma h}h^2 \gamma^2 - 4e^{-\gamma h}h \gamma}}{\gamma} & 0 \\ \frac{e^{-\gamma h} (\gamma h)^2}{\sqrt{4(1-e^{-\gamma h})-2e^{-\gamma h}h^2 \gamma^2 - 4e^{-\gamma h}h \gamma}} & \frac{\sqrt{2((1-e^{-\gamma h})^2-e^{-\gamma h}h^2\gamma^2)}}{\sqrt{2(1-e^{-\gamma h})-e^{-\gamma h}h^2 \gamma^2 - 2e^{-\gamma h}h \gamma}}
\end{pmatrix}.
\end{align*}
\end{proposition}

\begin{proof}
The proof is given in Section~\ref{sec:proof_num}.
\end{proof}

\begin{remark} If $K$ is symmetric positive definite but Assumption~\ref{ass:diag} is not satisfied, there still exists an orthogonal matrix $Q$ and a diagonal matrix $D=\mathrm{diag}(k_1,...,k_d)$ with $k_j>0$ such that $K=Q D Q^T$. Then, 
$\mathbf{A}(h)$ in \eqref{eq:stepG} is given by $P e^{\tilde{A}h} P^T$ with $P=\begin{pmatrix} Q & 0 \\ 0 & Q \end{pmatrix},$ and $\tilde{A}=\begin{pmatrix} 0 & I_d \\ -D & - \gamma I_d \end{pmatrix}$.
To calculate $\mathbf{B}(h)$ in \eqref{eq:stepG}, one needs to compute a representation of $\int_0^h Pe^{\tilde{A}(h-s)} P^T \rmd (0,B_s)^T$ via the two independent random variables $\xi, \zeta\sim \mathcal{N}(0_d,I_d)$ in the same spirit as in the proof of Proposition~\ref{prop:numscheme}.
\end{remark}

Then, for $h>0$ and two independent sequences $(\xi_k)_{k\in\mathbb{N}}$, $(\zeta_k)_{k\in\mathbb{N}}$ of  $d$-dimensional standard normally distributed random variables, the Markov chain $(\mathbf{X}_{k},\mathbf{V}_{k})_{k\in\mathbb{N}}$ corresponding to the $\mathcal{PG}$-splitting is given by
\begin{align} \label{eq:GP}
(\mathbf{X}_{k+1},\mathbf{V}_{k+1}) & =\mathcal{P}\mathcal{G}((\mathbf{X}_k,\mathbf{V}_k), h, \xi_{k+1},\zeta_{k+1}) 
 =\mathcal{P}(\mathcal{G}(\mathbf{X}_k,\mathbf{V}_k, h, \xi_{k+1},\zeta_{k+1}),h).
\end{align}
The Markov chain $(\mathbf{X}_{k},\mathbf{V}_{k})_{k\in\mathbb{N}}$ corresponding to the symmetric $\mathcal{PGP}$-splitting is given by 
\begin{align} \label{eq:PGP}
(\mathbf{X}_{k+1},\mathbf{V}_{k+1}) & =\mathcal{P}\mathcal{G}\mathcal{P}((\mathbf{X}_k,\mathbf{V}_k), h, \xi_{k+1},\zeta_{k+1}) \nonumber
\\ & =\mathcal{P}(\mathcal{G}(\mathcal{P}(\mathbf{X}_k,\mathbf{V}_k,h/2), h, \xi_{k+1},\zeta_{k+1}),h/2).
\end{align}

\section{Convergence results} \label{sec:conv_res}


We provide convergence results in $L^2$-Wasserstein distance for the kinetic Langevin samplers corresponding to the $\mathcal{PG}$-splitting \eqref{eq:GP} and the $\mathcal{PGP}$-splitting \eqref{eq:PGP}.

We define the (twisted) $L^2$-Wasserstein on $\mathbb{R}^{2d}$ in the following way. Let $\rho: \mathbb{R}^{2d}\times \mathbb{R}^{2d}\to [0,\infty)$ be a distance of the form 
\begin{align} \label{eq:dist}
\rho((x,v),((x',v'))=\sqrt{ (x-x',v-v')\begin{pmatrix}
A & B \\B & C
\end{pmatrix}\begin{pmatrix}
x-x' \\ v-v'
\end{pmatrix}},
\end{align}
where 
$\begin{pmatrix}
A & B \\ B & C
\end{pmatrix}$ is a symmetric positive definite matrix.  Let $\nu,\eta\in\mathcal{P}_2(\mathbb{R}^{2d})$. Then, the $L^2$ Wasserstein distance is defined by 
\begin{align*}
\mathcal{W}_{2,\rho}(\nu,\eta)=\Big(\inf_{\omega\in \Gamma(\nu,\eta)}\int_{\mathbb{R}^{4d}} \rho((x,v),((x',v'))^2 \omega(\rmd x \rmd v \rmd x' \rmd v')\Big)^{1/2},
\end{align*}  
where $\Gamma(\nu,\eta)$ denotes the set of all couplings of $\nu,\eta$ on $\mathbb{R}^{4d}$. For $A,C=1_d$ and $B=0_d$, $\rho$ is the standard Euclidean distance and we obtain the standard $L^2$ Wasserstein distance which we denote by $\mathcal{W}_2$.

In the following we consider the distance introduced in \eqref{eq:dist} with
\begin{align} \label{eq:rho}
A=\gamma^{-2} K+\frac{(1-2\tau)^2}{2}1_d,\qquad B=\frac{1-2\tau}{2}1_d, \qquad C=\gamma^{-2}1_d
\end{align}
and
\begin{align} \label{eq:tau}
\tau=\min\Big(\frac{1}{8}, \frac{\kappa\gamma^{-2} }{4}\Big).
\end{align}

\begin{theorem}[Convergence in $L^2$ Wasserstein distance for the $\mathcal{PG}$-sampler] \label{thm:conv} 
Suppose Assumption~\ref{ass:U} holds true. Let $\nu_0,\eta_0\in \mathcal{P}_2(\mathbb{R}^{2d})$. Denote by $\nu_k$ (resp. $\eta_k$) the distribution of the resulting Markov chain given by \eqref{eq:GP} with initial distribution $\nu_0$ (resp. $\eta_0$). Assume
\begin{align} \label{eq:condh}
L_G \gamma^{-2} \le \frac{1}{2} \qquad \text{and} \qquad h \le \min\Big(\frac{1}{2\gamma},\frac{\gamma}{2L_K}\Big).
\end{align}
Then, for $k\in\mathbb{N}$ it holds
\begin{align} \label{eq:convW2}
\mathcal{W}_{2,\rho}(\nu_k, \eta_k)\le e^{-c kh}\mathcal{W}_{2,\rho}(\nu_0, \eta_0), \qquad \text{and} \qquad \mathcal{W}_{2}(\nu_k, \eta_k)\le M_1 e^{-c kh}\mathcal{W}_{2}(\nu_0, \eta_0),
\end{align}
where the contraction rate $c>0$ and the constant $M_1$ satisfy
\begin{align}
& c=\min\Big(\frac{\gamma}{8},\frac{\kappa \gamma^{-1}}{4}\Big), \qquad \text{and} \label{eq:c}
\\ & M_1= \Big(\frac{\max(L_K\gamma^{-2}+1,(3/2)\gamma^{-2})}{\min((1/4)\gamma^{-2},(9/128)+\kappa \gamma^{-2})}\Big)^{1/2}. \label{eq:M}
\end{align}
Further, there exists a unique invariant measure $\mu_h$ for the sampler corresponding to the $\mathcal{PG}$-splitting \eqref{eq:GP} and for $k\in\mathbb{N}$,
\begin{align} \label{eq:conv_inv}
\mathcal{W}_{2,\rho}(\nu_k, \mu_h)\le e^{-c kh}\mathcal{W}_{2,\rho}(\nu_0, \mu_h), \qquad \text{and} \qquad \mathcal{W}_{2}(\nu_k, \mu_h)\le M_1 e^{-c kh}\mathcal{W}_{2}(\nu_0, \mu_h),
\end{align}

\end{theorem}

\begin{proof}
The proof is given in Section~\ref{sec:proof_conv}.
\end{proof}

%
%
%

\begin{remark} We expect that convergence in $L^1$ Wasserstein distance holds for non-convex potentials satisfying analogous assumption as in \cite{Sc2024,ScWh2024}. Since we focused here on the dependence on the condition number and dimension in the complexity analysis, the non-convex potentials are not in the scope of this work.
\end{remark}

Similarly, convergence in $L^2$ Wasserstein distance holds for the $\mathcal{PGP}$-sampler.


\begin{theorem}[Convergence in $L^2$ Wasserstein distance for the $\mathcal{PGP}$-sampler] \label{thm:conv_PGP2} 
Let $\nu_0,\eta_0\in \mathcal{P}_2(\mathbb{R}^{2d})$. Suppose Assumption~\ref{ass:U} holds true. Denote by $\nu_k$ (resp. $\eta_k$) the distribution of the resulting Markov chain given by \eqref{eq:PGP} with initial distribution $\nu_0$ (resp. $\eta_0$). Assume 
\begin{align} \label{eq:condh2}
L_G\gamma^{-2} \le \frac{1}{2}  \qquad \text{and} \qquad  h\le \min\Big(\frac{1}{4\gamma},\frac{\gamma}{4 L_K}\Big).
\end{align}
Then, for $k\in\mathbb{N}$ it holds
\begin{align*}
\mathcal{W}_{2,\rho}(\nu_k, \eta_k)\le e^{-c kh}\mathcal{W}_{2,\rho}(\nu_0, \eta_0), \qquad \text{and} \qquad \mathcal{W}_{2}(\nu_k, \eta_k)\le M_1 e^{-c kh}\mathcal{W}_{2}(\nu_0, \eta_0),
\end{align*}
where the contraction rate $c>0$ and the constant $M_1$ satisfy \eqref{eq:c} and \eqref{eq:M}, respectively.
Further, there exists a unique invariant measure $\tilde{\mu}_h$ for the sampler corresponding to  the $\mathcal{PGP}$-splitting and for $k\in\mathbb{N}$,
\begin{align*}
\mathcal{W}_{2,\rho}(\nu_k, \tilde{\mu}_h)\le e^{-c kh}\mathcal{W}_{2,\rho}(\nu_0, \tilde{\mu}_h), \qquad \text{and} \qquad  \mathcal{W}_{2}(\nu_k, \tilde{\mu}_h)\le  M_1 e^{-c kh}\mathcal{W}_{2}(\nu_0, \tilde{\mu}_h).
\end{align*}
\end{theorem}

\begin{proof}
The proof is given in Section~\ref{sec:proof_conv}.
\end{proof}

\begin{remark} We observe that under the condition of Theorem~\ref{thm:conv}, we can alternatively prove for $k\in\mathbb{N}$
\begin{align*}
\mathcal{W}_{2,\rho}(\nu_k, \eta_k)\le \mathcal{C} e^{-c kh}\mathcal{W}_{2,\rho}(\nu_0, \eta_0), \qquad \text{and} \qquad \mathcal{W}_{2}(\nu_k, \eta_k)\le \mathcal{C} M_1 e^{-c kh}\mathcal{W}_{2}(\nu_0, \eta_0),
\end{align*}
where $\mathcal{C}=\sqrt{1+\gamma h}\le \sqrt{3/2}$. This holds true since the splitting scheme satisfies
\begin{align*}
(\mathcal{P}_{1/2}\mathcal{G}\mathcal{P}_{1/2})^k=(\mathcal{P}_{1/2}\mathcal{G})(\mathcal{P}\mathcal{G})^{k-1}\mathcal{P}_{1/2},
\end{align*}
where $\mathcal{P}_{1/2}$ denotes a half step. Therefore, using the contraction result of Theorem~\ref{thm:conv} for the steps $(\mathcal{P}\mathcal{G})^{k-1}$ it remains to control the steps $\mathcal{P}_{1/2}\mathcal{G}$ and $\mathcal{P}_{1/2}$. For $\mathcal{P}_{1/2}\mathcal{G}$ we can show contraction with the same rate as for $\mathcal{P}\mathcal{G}$, while the step $\mathcal{P}_{1/2}$ induces the error given by the constant $\mathcal{C}$.
\end{remark}

\section{Complexity analysis} \label{sec:compl_res}


We establish strong accuracy bounds and complexity bounds for the kinetic Langevin samplers corresponding to the $\mathcal{PG}$-splitting \eqref{eq:GP} and the $\mathcal{PGP}$-splitting \eqref{eq:PGP}.

\begin{theorem}[Strong accuracy for the $\mathcal{PG}$-sampler] \label{thm:strongacc}
Suppose Assumption~\ref{ass:U} and \eqref{eq:condh} hold true. 
Denote by $\mu$ the invariant measure of the exact continuous Langevin dynamics given by \eqref{eq:LD} and by $\mu_h$ the invariant measure of the kinetic Langevin sampler given by \eqref{eq:GP}. 
Further, we assume that $\int_{\mathbb{R}^{2d}} x \mu(\rmd x)=0$ and $\nabla G(0)=0$. Then, 
\begin{align*}
\mathcal{W}_{2,\rho}(\mu, \mu_h)&\le h 8 c^{-1}    \sqrt{d} L_G \gamma^{-1} \sqrt{\frac{2\gamma^2}{\kappa} + \frac{L_K}{\kappa} },
\end{align*}
where $c$ is given in \eqref{eq:c}.

\end{theorem}

\begin{proof}
The proof is given in Section~\ref{sec:proof_compl}.
\end{proof}

\begin{remark} The conditions $x^*=\int_{\mathbb{R}^{2d}} x \mu(\rmd x \rmd v)=0$ and $\nabla G(0)=0$ in Theorem~\ref{thm:strongacc} are technical assumptions, which in general can be removed. Deleting the first one leads to the same bound up to a constant and $d$ being replaced by $d+\kappa x^*$. Without the second condition the estimate $|\nabla G(x)| \le L_G|x|$ changes to $|\nabla G(x)|\le L_G|x|+ |\nabla G(0)|$ which instead of $L_G\sqrt{d}$ leads to $\sqrt{2}(L_G\sqrt{d}+|\nabla G(0)|)$ in the result.
\end{remark}

\begin{remark} Suppose the assumptions of Theorem~\ref{thm:strongacc} hold true. If $L_G\le \kappa$, we assume $\kappa\gamma^{-2}=(1/2)$. Then,
\begin{align*}
\mathcal{W}_{2,\rho}(\mu, \mu_h)&\le 64 h \sqrt{d} L_G \gamma^{-2} \sqrt{\frac{5L_K}{\kappa} }=  h \sqrt{d} (\sqrt{5}\cdot 32)\frac{L_G}{\kappa} \sqrt{\frac{L_K}{\kappa} }.
\end{align*}
If $L_G>\kappa$, we assume $L_G\gamma^{-2} =(1/2)$. Then, $c=\kappa\gamma^{-1}/4$ and
\begin{align*}
\mathcal{W}_{2,\rho}(\mu, \mu_h)&\le h  \frac{32 L_G }{\kappa}    \sqrt{d}  \sqrt{\frac{4 L_G+L_K}{\kappa}}.
\end{align*}

\end{remark}

\begin{theorem}[Complexity result the $\mathcal{PG}$-sampler] \label{thm:complex}
Let $\nu_0\in\mathcal{P}_2(\mathbb{R}^{2d})$. Suppose Assumption~\ref{ass:U} and \eqref{eq:condh} hold true. Further, we assume that $\int_{\mathbb{R}^{2d}} x \mu(\rmd x)=0$ and $\nabla G(0)=0$. Denote by $\nu_k$ the distribution of the resulting Markov chain given by \eqref{eq:GP} with initial distribution $\nu_0$. Then, for $k\in\mathbb{N}$, 
\begin{align*}
&\mathcal{W}_{2,\rho}(\nu_k, \mu)\le e^{-c kh}\mathcal{W}_{2,\rho}(\nu_0, \mu_h)+ h 8 c^{-1}    \sqrt{d} L_G \gamma^{-1} \sqrt{\frac{2\gamma^2}{\kappa} + \frac{L_K}{\kappa} },  \qquad \text{and}
\\ & \mathcal{W}_{2}(\nu_k, \mu)\le M_1 e^{-c kh}\mathcal{W}_{2}(\nu_0, \mu_h)+ h M_2 \Big(8 c^{-1}    \sqrt{d} L_G \gamma^{-1} \sqrt{\frac{2\gamma^2}{\kappa} + \frac{L_K}{\kappa} }\Big),
\end{align*}
where $c$ and $M_1$ are given in \eqref{eq:c} and \eqref{eq:M}, respectively, and
\begin{align} \label{eq:M2}
M_2=\max\Big(2\gamma, \frac{1}{\sqrt{(9/128)+\kappa \gamma^{-2}}} \Big).
\end{align}

\end{theorem}

\begin{proof}
The proof is given in Section~\ref{sec:proof_compl}.
\end{proof}

\begin{remark}[$\varepsilon$-accuracy]
Assume that we can choose $\gamma$ such that the contraction rate $c$ is optimized, i.e., $\gamma^2=2\kappa$. Note that in this case by \eqref{eq:condh}, $L_G/\kappa\le 1$.
Then, for some $\varepsilon>0$ we obtain $\varepsilon$-accuracy, i.e., $\mathcal{W}_{2,\rho}(\nu_k, \mu)\le \varepsilon$, by setting
\begin{align*}
h^{-1}\ge 64 \sqrt{d} \frac{L_G}{\kappa}\sqrt{5 \frac{L_K}{\kappa}}\varepsilon^{-1}
\end{align*}
and the number of steps by 
\begin{align*}
k\ge h^{-1} \frac{8}{\sqrt{2\kappa}} \log\Big(\frac{2M_1 \mathcal{W}_{2}(\nu_0, \mu_h)}{\varepsilon}\Big) \ge  \sqrt{d} \frac{L_G}{\kappa}\sqrt{5 \frac{L_K}{\kappa}}\varepsilon^{-1}\frac{2^9}{\sqrt{2\kappa}} \log\Big(\frac{2M_1 \mathcal{W}_{2}(\nu_0, \mu_h)}{\varepsilon}\Big).
\end{align*}
Hence, $h^{-1}$ is of order $\mathcal{O}(\sqrt{d}L_K/\kappa\varepsilon^{-1})$.
Since there is one gradient evaluation of the function $G$ per step, we need $\mathcal{O}(\sqrt{dL_K}L_G/\kappa^2)$ gradient evaluations to obtain $\varepsilon$-accuracy.

\end{remark}

To establish strong accuracy bounds for the $\mathcal{PGP}$-sampler, we assume additionally:

\begin{assumption}\label{ass_GHess}
Suppose that $G$ is three times continuously differentiable. Let $\nabla^3 G$ denote the tensor of third derivatives that for each $x\in\mathbb{R}^d$ is a bilinear operator mapping $(v,w)\in\mathbb{R}^d$ to $\mathbb{R}^d$.
Suppose that there exists a constant $L_H<\infty$ such that for all $v,w\in\mathbb{R}^d$
\begin{align*}
\sup_{x\in\mathbb{R}^d}\|\nabla^3 G(x)[v,w]\| \le L_H|v||w|.
\end{align*}
\end{assumption}

\begin{theorem}[Strong accuracy for the $\mathcal{PGP}$-sampler] \label{thm:stracc_PGP}
Suppose Assumption~\ref{ass:U}, Assumption~\ref{ass_GHess} and \eqref{eq:condh2} hold true. 
Denote by $\mu$ the invariant measure of the exact continuous Langevin dynamics given by \eqref{eq:LD} and by $\tilde{\mu}_h$ the invariant measure of the kinetic Langevin sampler given by \eqref{eq:PGP}. 
Further, we assume that $\int_{\mathbb{R}^{2d}} x \mu(\rmd x)=0$ and $\nabla G(0)=0$. Then, 
\begin{align*}
\mathcal{W}_{2,\rho}(\mu, \tilde{\mu}_h)&\le h^2 \frac{1}{c}\mathcal{C}\Big(L_G \sqrt{d}\sqrt{ \frac{L_K}{\gamma^2}+ \frac{\gamma^2}{\kappa}+ \frac{L_K}{\kappa}+\frac{L_K^2}{\kappa \gamma^2}}+ \sqrt{L_H}\frac{d}{\gamma}\Big),
\end{align*}
where $c$ is given in \eqref{eq:c} and $\mathcal{C}\in\mathbb{R}_+$ is some number.
\end{theorem}

\begin{proof}
The proof is given in Section~\ref{sec:proof_compl}.
\end{proof}

\begin{theorem}[Complexity result for the $\mathcal{PGP}$-sampler] \label{thm:complex_PGP}

Let $\nu_0\in\mathcal{P}_2(\mathbb{R}^{2d})$. Suppose Assumption~\ref{ass:U}, Assumption~\ref{ass_GHess} and \eqref{eq:condh2} hold true. Further, we assume that $\int_{\mathbb{R}^{2d}} x \mu(\rmd x)=0$ and $\nabla G(0)=0$. Denote by $\nu_k$ the distribution of the resulting Markov chain given by \eqref{eq:PGP} with initial distribution $\nu_0$. Then, for $k\in\mathbb{N}$, 
\begin{align*}
\mathcal{W}_{2,\rho}(\nu_k, \mu)&\le e^{-c kh}\mathcal{W}_{2,\rho}(\nu_0, \tilde{\mu}_h)+  \frac{h^2}{c}\mathcal{C}\Big(L_G \sqrt{d}\sqrt{ \frac{L_K}{\gamma^2}+ \frac{\gamma^2}{\kappa}+ \frac{L_K}{\kappa}+\frac{L_K^2}{\kappa \gamma^2}}+ \sqrt{L_H}\frac{d}{\gamma}\Big),
\\  \mathcal{W}_{2}(\nu_k, \mu)&\le M_1 e^{-c kh}\mathcal{W}_{2}(\nu_0, \tilde{\mu}_h)\\ & + h^2  \frac{M_2}{c}\mathcal{C}\Big(L_G \sqrt{d}\sqrt{ \frac{L_K}{\gamma^2}+ \frac{\gamma^2}{\kappa}+ \frac{L_K}{\kappa}+\frac{L_K^2}{\kappa \gamma^2}}+ \sqrt{L_H}\frac{d}{\gamma}\Big),
\end{align*}
where $c$, $M_1$ and $M_2$ are given in \eqref{eq:c}, \eqref{eq:M} and \eqref{eq:M2}, respectively, and $\mathcal{C}\in\mathbb{R}_+$ is some number.
\end{theorem}

\begin{proof}
The proof is given in Section~\ref{sec:proof_compl}.
\end{proof}

\begin{remark} 
For the strong accuracy bounds we observe: If $L_G\le \kappa$, we assume $\kappa\gamma^{-2}=(1/2)$. Then, $c=\gamma/8$ and
\begin{align*}
\mathcal{W}_{2,\rho}(\mu, \tilde{\mu}_h)&\le h^2 \frac{8}{\gamma}\mathcal{C}\Big(L_G \sqrt{ \frac{L_K}{2\kappa}+ \frac{2\kappa}{\kappa}+ \frac{L_K}{\kappa}+\frac{L_K^2}{2\kappa^2}}\sqrt{d}+ \sqrt{L_H}\frac{d}{\gamma}\Big)
\\ & \le  h^2 \frac{8}{\sqrt{2\kappa}}\mathcal{C}\Big(L_G 2 \frac{L_K}{\kappa}\sqrt{d}+ \sqrt{L_H}\frac{d}{\gamma}\Big).
\end{align*}
If $L_G>\kappa$, we assume $L_G\gamma^{-2} =(1/2)$. Then, $c=\kappa\gamma^{-1}/4=\kappa/(4\sqrt{2L_G})$ and
\begin{align*}
\mathcal{W}_{2,\rho}(\mu, \tilde{\mu}_h)&\le h^2 \frac{4\sqrt{2L_G}}{\kappa}\mathcal{C}\Big(L_G \sqrt{d}\sqrt{ \frac{L_K}{2L_G}+ \frac{2L_G}{\kappa}+ \frac{L_K}{\kappa}+\frac{L_K^2}{\kappa 2L_G}}+ \sqrt{L_H}\frac{d}{\gamma}\Big)
\\ & \le \mathcal{O} \Big(h^2 \frac{1}{\sqrt{\kappa}}\Big(L_G \sqrt{d}\frac{L_K+L_G}{\kappa}+ \sqrt{L_H\frac{L_G}{\kappa}}\frac{d}{\gamma}\Big)\Big).
\end{align*}
Then, for some $\varepsilon>0$ we obtain $\varepsilon$-accuracy, i.e., $\mathcal{W}_{2,\rho}(\nu_k, \mu)\le \varepsilon$, by taking $h^{-1}$ of order $\mathcal{O}(d^{1/4} \varepsilon^{-1/2})$. If $\exp(-chk) M_1 \mathcal{W}_2(\nu_0, \mu_h^*)\le \varepsilon/2$, we obtain for the number $k$ of steps the order $\mathcal{O}(d^{1/4} \varepsilon^{-1/2} \log(\varepsilon^{-1}))$.
\end{remark}

\begin{remark}[Comparison to other samplers and conditions on the potential]

We observe that the above dimension dependency in the complexity guarantees are comparable with the one of OBABO \cite{LePaWh2024a} and the UBU splitting \cite{SaZy2021,PaWh2024}.

Further, we note that often stricter conditions on the Hessian of $ G$ can be imposed. For a broad range of applications $G$ satisfies the strongly Hessian Lipschitz assumption as in \cite[Assumption 2]{PaWh2024} (see also \cite{ChGaJi2023}).
Assuming $G$ to be strongly Hessian Lipschitz, we can obtain a better dimension dependence in the strong accuracy estimate and in the complexity bounds.
In particular, using \cite[Lemma 7]{PaWh2024} the terms \eqref{eq:badterm1} and \eqref{eq:badterm2} in the proof of Theorem~\ref{thm:complex_PGP}, which give the overall dependence $\mathcal{O}(h^2d)$, can be bounded by terms of order $\mathcal{O}(h^6 d)$ instead of $\mathcal{O}(h^6 d^2)$. This estimate improves the bound in Theorem~\ref{thm:complex_PGP} to $\mathcal{O}(h^2 \sqrt{d})$.
\end{remark}

\section{Proofs} \label{sec:proofs}

\subsection{Proofs of the convergence results} \label{sec:proof_conv}

\begin{proof}[Proof of Theorem~\ref{thm:conv}]
Fix $h>0$ satisfying \eqref{eq:condh}.
Let $(x,v),(x',v')\in\mathbb{R}^{2d}$ and $\xi,\zeta\sim \mathcal{N}(0_d,I_d)$ be two independent standard normally distributed random variables. In the following we apply a synchronously coupled transition step of $\mathcal{P}\mathcal{G}$-step to $(x,v)$ and $(x',v')$, i.e., we apply the same pair of standard normally distributed random variables $(\xi, \zeta)$ in the transition step \eqref{eq:GP},
\begin{align*}
((\mathbf{X},\mathbf{V}),(\mathbf{X}',\mathbf{V}'))=(\mathcal{P}\mathcal{G}(x,v,h,\xi, \zeta),\mathcal{P}\mathcal{G}(x',v',h,\xi, \zeta)).
\end{align*}
For the $\mathcal{G}$-step it holds
\begin{align*}
((\mathbf{X}^G,\mathbf{V}^G),(\mathbf{X}'^G,\mathbf{V}'^G))=(\mathcal{G}(x,v,h,\xi, \zeta),\mathcal{G}(x',v',h,\xi, \zeta))=(x_h,v_h,x'_h,v_h'),
\end{align*}
where $(x_s,v_s,x'_s,v_s')_{s\in[0,h]}$ is given as a solution to 
\begin{align} \label{eq:exactGauss}
\begin{cases}
 \rmd x_s&= v_s\rmd s
\\ \rmd v_s&=(-\gamma v_s - K x_s)\rmd s + \sqrt{2 \gamma}\rmd B_s
\\ \rmd x_s'&= v_s'\rmd s
\\ \rmd v_s'&=(-\gamma v_s' - K x_s')\rmd s + \sqrt{2 \gamma}\rmd B_s
\end{cases}
\end{align}
and $(x_0,v_0,x'_0,v_0')=(x,v,x',v')$. Note that since we represent the $\mathcal{G}$-step by the continuous dynamics, the pair of random variables  $(\xi, \zeta)$ are replaced by the Brownian motion $(B_t)_{t\ge 0}$ in both copies resulting in a synchronous coupling for the continuous dynamics.
Then, the difference process $(z_s,w_s)_{s\in[0,h]}=(x_s-x_s',v_s-v_s')_{s\in[0,h]}$ satisfies
\begin{align} \label{eq:diffproc}
\begin{cases}
\frac{\rmd}{\rmd s}z_s=w_s,
\\ \frac{\rmd}{\rmd s}w_s=-\gamma w_s -K z_s,
\end{cases}
\end{align}
and we define 
\begin{align} \label{eq:diff_Gstep}
(\mathbf{Z}^G,\mathbf{W}^G)=(\mathbf{X}^G-\mathbf{X}'^G,\mathbf{V}^G-\mathbf{V}'^G).
\end{align}
For $\rho:\mathbb{R}^{2d}\times \mathbb{R}^{2d}\to [0,\infty)$, given by \eqref{eq:dist} and \eqref{eq:rho}, 
we write 
$\rho((x_s,v_s),( x_s', v_s'))=\rho(z_s,w_s)$ for simplicity and obtain analogously to \cite[Theorem 1]{Sc2024} by Ito's formula
\begin{align*}
\frac{\rmd }{\rmd s} \rho(z_s,w_s)^2&=-(1-2\tau)\gamma^{-1} z_s \cdot(K z_s) - (1-2\tau)\gamma^{-1}(2\tau \gamma) z_s\cdot w_s- \frac{1+2\tau}{\gamma} |w_s|^2
\\ & \le -2\tau \gamma \rho(z_s,w_s)^2- \frac{z_s\cdot(Kz_s)}{4\gamma}  - \gamma^{-1}|w_s|^2,
\end{align*}
where the last step holds since for all $z\in\mathbb{R}^d$
\begin{align*}
-(1-2 \tau)\gamma^{-1} z\cdot(Kz)&=-(1-4 \tau)\gamma^{-1} z\cdot(Kz)- 2\tau\gamma^{-1}z\cdot(Kz)
\\ & \le -\frac{z\cdot(Kz)}{2\gamma}- 2\tau\gamma^{-1}z\cdot(Kz)
\\ & \le -(2\tau \gamma)\Big(\frac{z\cdot(Kz)}{\gamma^{2}}+\frac{(1-2\tau)^2}{2}|z|^2\Big)-\frac{z\cdot(K z)}{4\gamma}
\end{align*}
due to applying \eqref{eq:tau} twice.
Hence, for all $t\in[0,h]$
\begin{align} \label{eq:Grho}
\rho(z_t, w_t)^2-\rho(z_0,w_0)^2\le \int_0^t (-2\tau \gamma) \rho(z_s,w_s)^2 \rmd s +\int_0^t \Big(- \frac{z_s\cdot(Kz_s)}{4\gamma}  - \gamma^{-1}|w_s|^2\Big)\rmd s.
\end{align}
Further, we observe by \eqref{eq:diffproc} and Young's inequality
\begin{align*}
|w_t|^2-|w_s|^2& =\int_s^t w_u\cdot(-2\gamma w_u-2K z_u)\rmd u 
\\ &\le \int_s^t -2\gamma|w_u|^2+ \gamma |w_u|^2 + \frac{|Kz_u|^2}{\gamma} \rmd u =
 \int_s^t-\gamma |w_u|^2+ \frac{|Kz_u|^2}{\gamma} \rmd u.
\end{align*}
Inserting this bound in \eqref{eq:Grho}, we obtain
\begin{align*}
\rho(z_t, w_t)^2-\rho(z_0,w_0)^2&\le \int_0^t (-2\tau \gamma) \rho(z_s,w_s)^2 \rmd s -\int_0^t \gamma^{-1}|w_t|^2\rmd s
\\ & +\int_0^t \Big(- \frac{z_s\cdot(Kz_s)}{\gamma}  + \int_s^t\Big(-|w_u|^2+\frac{|Kz_u|^2}{2\gamma^2} \Big)\rmd u\Big)\rmd s.
\end{align*}
For the two terms involving $K$ it holds
\begin{align*}
\int_0^t \Big(- \frac{z_s\cdot(Kz_s)}{4\gamma}+ \int_s^t\frac{|Kz_u|^2}{2\gamma^2} \rmd u\Big)\rmd s&\le \int_0^t \Big(- \frac{z_s\cdot(Kz_s)}{4\gamma}+ h\frac{|Kz_s|^2}{2\gamma^2}\Big)\rmd s
\\ & \le \int_0^t \Big(- \frac{z_s\cdot(Kz_s)}{4\gamma}+ h\frac{L_K z_s(Kz_s)}{2\gamma^2}\Big)\rmd s\le 0,
\end{align*}
where we applied \cite[Theorem 2.1.5]{Ne2018} in the second last step and \eqref{eq:condh} in the last one.
Further,
\begin{align*}
\int_0^t \int_s^t |w_u|^2 \rmd u \rmd s &= \int_0^t \int_0^t 1_{\{s\le u\} } |w_u|^2 \rmd u \rmd s= \int_0^t \int_0^t 1_{\{s\le u\} } |w_u|^2 \rmd s \rmd u
\\ & =\int_0^t |w_u|^2 \int_0^u  \rmd s \rmd u=\int_0^t |w_u|^2 u \rmd u.
\end{align*}
Hence, for any $t\in[0,h]$
\begin{align*}
\rho(z_t, w_t)^2-\rho(z_0,w_0)^2&\le \int_0^t (-2\tau \gamma) \rho(z_s,w_s)^2 \rmd s -t  \gamma^{-1}|w_t|^2- \int_0^t |w_s|^2 s \rmd s
\\ & \le  \int_0^t (-2\tau \gamma) (\rho(z_s,w_s)^2+s\gamma^{-1}|w_s|^2) \rmd s -t  \gamma^{-1}|w_t|^2,
\end{align*}
where we applied $\tau\le (1/2)$ in the last step. Hence the continuous function $f(s)=\rho(z_s, w_s)^2+ s\gamma^{-1} |w_s|^2$ 
satisfies
\begin{align*}
f(t)-f(0)\le \int_0^t (-2\tau \gamma) f(s) \rmd s.
\end{align*}
By Grönwall's inequality we obtain
\begin{align*}
f(t)\le e^{-2\gamma \tau t} f(0)
\end{align*}
and thus for $t=h$
\begin{align} \label{eq:Gstep_conv}
\rho((\mathcal{G}(x,v,h,\xi, \zeta),\mathcal{G}(x',v',h,\xi, \zeta))^2\le e^{-2 \gamma \tau h}\rho((x,v),(x',v'))^2-h\gamma^{-1}|\mathbf{W}^G|^2.
\end{align}
with $\mathbf{W}^G$ given by \eqref{eq:diff_Gstep}.
Further, we observe
\begin{align*}
\rho((\mathcal{P}(x,v,h),&\mathcal{P}(x',v',h))^2
\\ &\le \rho((x,v),(x',v'))^2-h (1-2\tau)\gamma^{-1}(\nabla G(x)-\nabla G(x'))(x-x')
\\ & -2h\gamma^{-2}(\nabla G(x)-\nabla G(x'))(v-v')+ h^2\gamma^{-2}|\nabla G(x)-\nabla G(x')|^2.
\end{align*}
By Young's inequality, \eqref{eq:tau} and \cite[Theorem 2.1.5]{Ne2018} we obtain
\begin{align*}
&-h (1-2\tau)\gamma^{-1}(\nabla G(x)-\nabla G(x'))(x-x') -2h\gamma^{-2}(\nabla G(x)-\nabla G(x'))(v-v')
\\ & \qquad  + h^2\gamma^{-2}|\nabla G(x)-\nabla G(x')|^2
\\ & \le \gamma^{-1} h  |v-v'|^2 + h\gamma^{-1} \Big(-\frac{3}{4}+\gamma^{-2} L_G+ h\gamma^{-1} L_G\Big)(\nabla G(x)-\nabla G(x'))(x-x')\le 0,
\end{align*}
where the last step holds by convexity of $G$ and the fact that
the prefactor is non-positive due to \eqref{eq:condh}.
Hence, 
\begin{align*}
\rho((\mathcal{P}(x,v,h),\mathcal{P}(x',v',h))^2\le  \rho((x,v),(x',v'))^2+\gamma^{-1} h |v-v'|^2.
\end{align*}
Combining the two steps we obtain
\begin{equation}\label{eq:contrPG}
\begin{aligned} 
\rho((\mathbf{X},\mathbf{V}),(\mathbf{X}',\mathbf{V}'))^2&=
\rho((\mathcal{P}\mathcal{G}(x,v,h,\xi, \zeta),\mathcal{P}\mathcal{G}(x',v',h,\xi, \zeta))^2
\\ & \le  \rho((\mathbf{X}^G,\mathbf{V}^G),(\mathbf{X}'^G,\mathbf{V}'^G))^2+\gamma^{-1} h |\mathbf{W}^G|^2
\\ & \le e^{-2\gamma \tau h}\rho((x,v),(x',v')).
\end{aligned}
\end{equation}
with $\mathbf{W}^G$ given by \eqref{eq:diff_Gstep}.
Taking expectation and square roots yields
\begin{align*}
\mathcal{W}_{2,\rho}(\nu_1,\eta_1)&\le \mathbb{E}[\rho((\mathbf{X},\mathbf{V}),(\mathbf{X}',\mathbf{V}'))^2]^{1/2}
\\ & \le  e^{-\gamma \tau h}\mathbb{E}_{(x,v)\sim \nu_0, (x',v')\sim \eta_0}[\rho((x,v),(x',v'))^2]^{1/2}.
\end{align*}
Taking the infimum over all couplings between $\eta_0$ and $\nu_0$, we obtain the first bound.

By equivalence of $\rho$ and the Euclidean distance in $\mathbb{R}^{2d}$, i.e., by \eqref{eq:tau}
\begin{equation}\label{eq:equivdist}
\begin{aligned} 
\min\Big(\frac{1}{4\gamma^2},\frac{9}{128}+\kappa \gamma^{-2}\Big)|(x,v)-(x',v')|^2&\le \rho((x,v),(x',v'))^2
\\ & \le \max\Big(L_K\gamma^{-2}+1,\frac{3}{2\gamma ^{2}}\Big) |(x,v)-(x',v')|^2
\end{aligned}
\end{equation}
we obtain the second bound.
The result \eqref{eq:conv_inv} is an immediate consequence of \eqref{eq:convW2} by setting $\eta_0=\mu_h$.  Note that the existence of the unique invariant measure follows by Banach fixed point theorem.
\end{proof}

\begin{proof}[Proof of Theorem~\ref{thm:conv_PGP2}]
As in the proof of Theorem~\ref{thm:conv}, we consider the process $(z_s,w_s)_{s\ge 0}$ given by \eqref{eq:diffproc} for the step $\mathcal{G}$ and it holds \eqref{eq:Grho}.
By the trapezoidal rule,
\begin{align*}
\int_0^t \frac{1}{2}&(|w_t|^2+|w_0|^2)-|w_s|^2\rmd s 
\\ & = \int_0^t \Big( \frac{1}{2} \int_0^t (t-r) \frac{\rmd^2}{\rmd r^2} |w_r|^2 \rmd r- \int_0^s (s-r) \frac{\rmd^2}{\rmd r^2} |w_r|^2 \rmd r \Big)\rmd s.
\end{align*}
By \eqref{eq:diffproc} and Young's inequality it holds
\begin{align*}
\frac{\rmd^2}{\rmd r^2} |w_r|^2&= \frac{\rmd}{\rmd r} (2 w_r\cdot(-\gamma w_r - K z_r))
\\ & = 4 \gamma^2 |w_r|^2 + 6\gamma w_r \cdot (K z_r) +2 (K z_r)^2- 2 w_r\cdot(Kw_r)  
\\ & \le 7 \gamma^{2} |w_r|^2 + 5 (Kz_r)^2 - 2 w_r\cdot(Kw_r), \qquad \text{and }
\\ \frac{\rmd^2}{\rmd r^2} |w_r|^2 & \ge - (1/4) (Kz_r)^2 - 2 w_r\cdot(Kw_r).
\end{align*}
Hence, for $t\le h$
\begin{align} \label{eq:boundtrapez}
\int_0^t  \frac{1}{2}(|w_t|^2&+|w_0|^2)-|w_s|^2\rmd s \nonumber
\\ & \le \int_0^t \frac{1}{2} \int_0^t (t-r) \Big(7 \gamma^{2} |w_r|^2 + 5 (Kz_r)^2 - 2 w_r\cdot(Kw_r)\Big) \rmd r \rmd s \nonumber
\\ & \quad - \int_0^t \int_0^s (s-r) \Big( - (1/4) (Kz_r)^2 - 2 w_r\cdot(Kw_r)\Big) \rmd r \rmd s \nonumber
\\ & \le \int_0^t \frac{1}{2} \int_0^t (t-r) \Big(7 \gamma^{2} |w_r|^2 + 5 (Kz_r)^2 \Big) \rmd r \rmd s \nonumber
\\ & \quad - \frac{1}{2}\int_0^t \int_0^s (s-r) \Big( - (1/2) (Kz_r)^2 - 2 w_r\cdot(Kw_r)\Big) \rmd r \rmd s \nonumber
\\ & \le \frac{t^2}{2} \int_0^t  \Big(7 \gamma^{2} |w_r|^2 + 5 (Kz_r)^2 + (1/2) (Kz_r)^2 +2 w_r\cdot(Kw_r)\Big) \rmd r\nonumber
\\ & \le \frac{h^2}{2} \int_0^t  \Big(7 \gamma^{2} |w_r|^2 +(11/2) L_K z_r \cdot(Kz_r) +2 L_K |w_r|^2\Big) \rmd r,
\end{align}
where we used Assumption~\ref{ass:U} and \cite[Theorem 2.1.5]{Ne2018}. We consider $\delta_1\in[0,1]$ and
insert \eqref{eq:boundtrapez} in \eqref{eq:Grho} to obtain 
\begin{align*}
&\rho(z_t, w_t)^2-\rho(z_0,w_0)^2  \le \int_0^t (-2\tau \gamma) \rho(z_s,w_s)^2 \rmd s +\int_0^t \Big(- \frac{z_s\cdot(Kz_s)}{4\gamma}  - \gamma^{-1}|w_s|^2\Big)\rmd s
\\ & \le  \int_0^t (-2\tau \gamma) \rho(z_s,w_s)^2 \rmd s +\int_0^t \Big(- \frac{z_s\cdot(Kz_s)}{4\gamma}  - \gamma^{-1}(1-\delta_1)|w_s|^2\Big)\rmd s
\\ &  - \delta_1 \gamma^{-1} \int_0^t \frac{1}{2}(|w_t|^2+|w_0|^2)\rmd s 
 +  \frac{\delta_1 h^2}{2 \gamma} \int_0^t  \Big((7 \gamma^{2}+2L_K) |w_r|^2 + \frac{11}{2} L_K z_r \cdot(Kz_r) \Big) \rmd r
\\ & \le  \int_0^t (-2\tau \gamma) \rho(z_s,w_s)^2 \rmd s - \delta_1 \gamma^{-1} \frac{t}{2}(|w_t|^2+|w_0|^2)
\\ &  +\int_0^t \Big( \Big(\frac{11 \delta_1 L_K h^2}{4\gamma} -\frac{1}{4\gamma}\Big) z_r \cdot(Kz_r) + \gamma^{-1}\Big( \delta_1 h^2(\frac{7}{2} \gamma^{2}+L_K)-(1-\delta_1)\Big)|w_r|^2\Big)\rmd r.
\end{align*}
Then, there exists $\delta_1\in [0,1]$ such that the last integral on the right hand side is bounded by 
\begin{align} \label{eq:boundrhoG2}
\int_0^t \Big( \Big(\frac{11 \delta_1 L_K h^2}{4\gamma} & -\frac{1}{4\gamma}\Big) z_r \cdot(Kz_r) + \gamma^{-1}\Big( \delta_1 h^2((7/2) \gamma^{2}+L_K)-(1-\delta_1)\Big)|w_r|^2\Big)\rmd r \nonumber
\\ & \qquad  \le -\tau \gamma \delta_1 \int_0^t \gamma^{-1} s |w_s|^2 \rmd s.
\end{align} 
In particular, it holds true for $\delta_1=(3/4)$ since by \eqref{eq:condh2} and \eqref{eq:tau}
\begin{align*}
11 \delta_1 h^2 L_K\le 11 \delta_1 \frac{1}{16}\le 1 \qquad \text{and } \qquad \delta_1 h^2((7/2) \gamma^{2}+L_K)-(1-\delta_1)\le -\tau \gamma \delta_1  h,
\end{align*}
where the second inequality is satisfied since for $\delta_1=(3/4)$
\begin{align*}
\delta_1 \Big(h^2((\frac{7}{2}\gamma^2+L_K)+ \tau \gamma h+1\Big)\le \delta_1 \Big(\frac{9}{2}\frac{1}{16}+ \frac{1}{8} \frac{1}{4}+1\Big)\le 1.
\end{align*}
%
Further, we observe that for $t\le h$ by \eqref{eq:tau} and \eqref{eq:condh2}
\begin{align*}
-\frac{t}{2}\gamma^{-1} |w_0|^2 \le - \frac{t}{2}\gamma^{-1} 8\tau  |w_0|^2 \le - \frac{t^2}{2} 32\tau |w_0|^2 = (-\tau \gamma)\int_0^t 32 \gamma^{-1} s|w_0|^2 \rmd s  .
\end{align*}
Hence,
\begin{align*}
\rho(z_t, w_t)^2&+\delta_1 \gamma^{-1} \frac{t}{2}\Big(|w_t|^2+\frac{32}{33}|w_0|^2\Big)-\rho(z_0,w_0)^2  
\\ & \le \int_0^t (-2\tau \gamma) (\rho(z_s,w_s)^2 +\delta_1 \frac{\gamma^{-1}}{2}|w_s|^2)   \rmd s -\delta_1 \gamma^{-1} \frac{t}{2}\frac{1}{33}|w_0|^2
\\ & \le \int_0^t (-2\tau \gamma) \Big(\rho(z_s,w_s)^2 +\delta_1 \frac{\gamma^{-1}}{2}|w_s|^2 +\delta_1\frac{32}{33}  \frac{\gamma^{-1}}{2} s|w_0|^2\Big)   \rmd s 
\end{align*}
and therefore by Grönwall's inequality
\begin{align*}
\rho(z_t, w_t)^2&+\delta_1 \gamma^{-1} \frac{t}{2}\Big(|w_t|^2+\frac{32}{33}|w_0|^2\Big)\le e^{-2\gamma \tau t} \rho(z_0, w_0)^2.
\end{align*}
In particular for $t=h$ and $\delta_1=(3/4)$
\begin{align} \label{eq:contrG}
\rho(\mathcal{G}(x,v,h, \xi,\zeta),&\mathcal{G}(x',v',h, \xi,\zeta))^2 
\\ &   \le e^{-2\gamma \tau h} \rho((x,v),(x',v'))^2-\frac{3}{4} \gamma^{-1} \frac{h}{2}\Big(|\mathbf{W}^G|^2+\frac{32}{33}|v-v'|^2\Big) \nonumber
\\ & \le e^{-2\gamma \tau h} \rho((x,v),(x',v'))^2- \gamma^{-1} \frac{4h}{11}\Big(|\mathbf{W}^G|^2+|v-v'|^2\Big)
\end{align}
with $\mathbf{W}^G$ given by \eqref{eq:diff_Gstep}.
Next, we set $\delta_2=\frac{\delta_1}{2}\frac{32}{33}=\frac{4}{11}$.
For the $\mathcal{P}$-step we obtain
\begin{align*}
\rho&((\mathcal{P}(x,v,h/2),\mathcal{P}(x',v',h/2))^2
\\ & = \rho((x,v),(x',v'))^2-(h/2) (1-2\tau)\gamma^{-1}(\nabla G(x)-\nabla G(x'))(x-x')
\\ & \quad -h\gamma^{-2}(\nabla G(x)-\nabla G(x'))(v-v')+ (h^2/4)\gamma^{-2}|\nabla G(x)-\nabla G(x')|^2
\\ & \le \rho((x,v),(x',v'))^2-(3h/8)\gamma^{-1}(\nabla G(x)-\nabla G(x'))(x-x')
\\ & \quad +h \delta_2|v-v'|^2+ \Big(h\gamma^{-3} \delta_2^{-1}/4+ (h^2/4)\gamma^{-2}\Big)|\nabla G(x)-\nabla G(x')|^2
\\ & \le \rho((x,v),(x',v'))^2+ h \delta_2|v-v'|^2
\\ & \quad +\Big(-(3h/8)\gamma^{-1}+ L_G h\gamma^{-3} \delta_2^{-1}/4+ L_G h\gamma^{-3}(1/16) \Big)(\nabla G(x)-\nabla G(x'))(x-x'),
\end{align*}
where we applied Young's inequality with $\delta_2>0$ in the second step and \cite[Theorem 2.1.5]{Ne2018} and Assumption~\ref{ass:U} in the third step.
By \eqref{eq:condh2} and since $\delta_2=4/11$,
\begin{align*}
-(3h/8)\gamma^{-1}+ L_G h\gamma^{-3}\delta_2^{-1}/4+ L_G h\gamma^{-3}(1/16) \le h\gamma^{-1}\Big(-\frac{3}{8}+ \frac{1}{2}\frac{11}{16}+\frac{1}{2}\frac{1}{16}\Big)=0,
\end{align*}
which implies
\begin{align*}
\rho((\mathcal{P}(x,v,h/2),\mathcal{P}(x',v',h/2))^2 & \le \rho((x,v),(x',v'))^2+ h \frac{4}{11}|v-v'|^2.
\end{align*}
Combining this estimate for the $\mathcal{P}$-step with \eqref{eq:contrG} for the $\mathcal{G}$-step 
yields
\begin{align} 
\rho&( (\mathbf{X},\mathbf{V}),  (\mathbf{X}',\mathbf{V}') )^2 \nonumber
 =\rho(\mathcal{P}(\mathcal{G}(\mathcal{P}(x,v,\frac{h}{2}),h, \xi,\zeta),\frac{h}{2}),\mathcal{P}(\mathcal{G}(\mathcal{P}(x',v',\frac{h}{2}),h, \xi,\zeta),\frac{h}{2}))^2 \nonumber
\\ & \le \rho(\mathcal{G}\mathcal{P}_{1/2}(x,v),\mathcal{G}\mathcal{P}_{1/2}(x',v'))^2+ h \frac{4}{11}|\mathcal{G}\mathcal{P}_{1/2}(x,v)_v-\mathcal{G}\mathcal{P}_{1/2}(x',v')_v|^2 \nonumber
\\ & \le e^{-2\gamma \tau h} \rho(\mathcal{P}_{1/2}(x,v),\mathcal{P}_{1/2}(x',v'))^2-\delta_1 \gamma^{-1} \frac{h}{2}\frac{32}{33}\Big(|w_h|^2+|v-v'|^2\Big) \nonumber
\\ & + h \frac{4}{11}|\mathcal{G}\mathcal{P}_{1/2}(x,v)_v-\mathcal{G}\mathcal{P}_{1/2}(x',v')_v|^2 \nonumber
\\ & \le e^{-2\gamma \tau h}\Big(\rho((x,v),(x',v'))^2+ h \frac{4}{11}|v-v'|^2\Big)-\frac{4}{11}\Big(|w_h|^2+|v-v'|^2\Big) + h \frac{4}{11}|w_h|^2 \nonumber
\\ & \le e^{-2\gamma \tau h}\rho((x,v),(x',v'))^2, \nonumber
\end{align}
where $\mathcal{G}\mathcal{P}_{1/2}(x,v)_v$ denotes the second component of a $\mathcal{G}\mathcal{P}_{1/2}$ step with inital data $(x,v)$. Note that we used
$\mathcal{G}\mathcal{P}_{1/2}(x,v)_v-\mathcal{G}\mathcal{P}_{1/2}(x',v')_v=w_h$
where $w_h$ is given by \eqref{eq:diffproc} with initial data $(\mathcal{P}_{1/2}(x,v),\mathcal{P}_{1/2}(x',v'))$.
Taking expectation and square roots yields
\begin{align*}
\mathcal{W}_{2,\rho}(\nu_1,\eta_1) 
\le  e^{-\gamma \tau h}\mathbb{E}_{(x,v)\sim \nu_0, (x',v')\sim \eta_0}[\rho((x,v),(x',v'))^2]^{1/2}.
\end{align*}
Taking the infimum over all couplings between $\eta_0$ and $\nu_0$, we obtain the first bound and by \eqref{eq:equivdist} the second bound. The third and forth bounds are an immediate consequences by setting $\eta_0=\tilde{\mu}_h$ which uniquely exists by Banach fixed point theorem.
\end{proof}

\subsection{Proofs of the complexity results} \label{sec:proof_compl}


%


\begin{proof}[Proof of Theorem~\ref{thm:strongacc}]
Fix $h$ satisfying \eqref{eq:condh}. Consider $l\in\mathbb{N}$ whose precise value we fix later. Denote by $(p_t)_{t\ge 0}$ the transition function of the exact continuous dynamics and by $\pi_h$ the transition kernel of the kinetic Langevin sampler given by \eqref{eq:GP}. Since $\mu$ and $\mu_h$ are the invariant probability measures of the exact dynamics and the kinetic Langevin sampler, respectively, and it holds
\begin{align*}
\mathcal{W}_{2,\rho}(\mu, \mu_h)& = \mathcal{W}_{2,\rho}(\mu p_{lh}, \mu_h \pi_h^l)\le \mathcal{W}_{2,\rho}(\mu p_{lh}, \mu \pi_h^l)+\mathcal{W}_{2,\rho}(\mu \pi_{h}^l, \mu_h \pi_h^l)
\\ & \le \mathcal{W}_{2,\rho}(\mu p_{lh}, \mu \pi_h^l)+e^{-clh}\mathcal{W}_{2,\rho}(\mu , \mu_h ),
\end{align*}
where Theorem~\ref{thm:conv} is applied in the last step.
Hence, 
\begin{align} \label{eq:W2_stracc}
\mathcal{W}_{2,\rho}(\mu, \mu_h)&\le (1-e^{-clh})^{-1} \mathcal{W}_{2,\rho}(\mu p_{lh}, \mu \pi_h^l)\le \Big(1+\frac{1}{clh}\Big)\mathcal{W}_{2,\rho}(\mu p_{lh}, \mu \pi_h^l).
\end{align}
To bound the right hand side, we consider a synchronous coupling of the exact Langevin dynamics $(X_t,V_t)_{t\ge 0}$ given by \eqref{eq:LD} and the kinetic Langevin sampler given by $(\mathbf{X}_k,\mathbf{V}_k)_{k\in\mathbb{N}}$ with identical initial conditions $(X_0,V_0)=(\mathbf{X}_0,\mathbf{V}_0)=(x,v)$. 
We represent $({\mathbf{X}}_k,{\mathbf{V}}_k)_{k\in\mathbb{N}}$ by the following recursive scheme where the splitting steps are successively preformed:
\begin{align*}
\mathcal{G}:\qquad& ({\mathbf{X}}_{k+1}^G,{\mathbf{V}}_{k+1}^G)=(\tilde{X}_{(k+1)h}, \tilde{V}_{(k+1)h})
\\\mathcal{P}:\qquad & ({\mathbf{X}}_{k+1},{\mathbf{V}}_{k+1})=({\mathbf{X}}_{k+1}^G,{\mathbf{V}}_{k+1}^G-(h/2)\nabla G({\mathbf{X}}_{k+1}^G)),
\end{align*} 
where
\begin{align} \label{eq:exactG}
&\begin{cases}
\rmd \tilde{X}_s=\tilde{V}_s \rmd s
\\ \rmd \tilde{V}_s=(-\gamma \tilde{V}_s-K\tilde{X}_s) \rmd s+ \sqrt{2\gamma}\rmd B_s
\end{cases}
\text{with } \tilde{X}_{kh}={\mathbf{X}}_k, \qquad \tilde{V}_{kh}={\mathbf{V}}_k.
\end{align}
%
For the synchronous coupling, we consider the two processes on a joint probability space and take the same Brownian motion $(B_t)_{t\ge 0}$ in \eqref{eq:LD} and \eqref{eq:exactG}.
Then, it holds for $(\mathbf{Z}_k,\mathbf{W}_k)=(X_{kh}-\mathbf{X}_{k},V_{kh}-\mathbf{V}_{k})$
\begin{align*}
\begin{cases}
\mathbf{Z}_{k+1}=\mathbf{Z}_k+ \int_0^h W_{kh+r} \rmd r
\\  \mathbf{W}_{k+1}=\mathbf{W}_k- \int_0^h (\gamma W_{kh+r}+KZ_{kh+r}+\nabla G(X_{kh+r})-\nabla G(\mathbf{X}_{k+1})) \rmd r,
\end{cases}
\end{align*}
where $(Z_u,W_u)_{u\ge 0}=(X_{u}-\tilde{X}_u,V_{u}-\tilde{V}_u)_{u\ge 0}$ and $(\mathbf{Z}_0,\mathbf{W}_0)=(0,0)$. 
We define 
\begin{align}\label{eq:ab}
a_k=\Big(\frac{L_K}{2\gamma^2}+1\Big)\hat{a}_k=\Big(\frac{L_K}{2\gamma^2}+1\Big)\mathbb{E}_{(x,v)\sim \mu}[|\mathbf{Z}_k|^2], \qquad b_k=\mathbb{E}_{(x,v)\sim \mu}[|\mathbf{Z}_k+\gamma^{-1} \mathbf{W}_k|^2]
\end{align}
for $k\in\mathbb{N}$.
Then, 
\begin{equation} \label{eq:W2_stracc2}
\begin{aligned}
\mathcal{W}_{2,\rho}(\mu p_{lh},& \mu \pi_h^l)^2\le \mathbb{E}_{(x,v)\sim \mu}\Big[\begin{pmatrix}
\mathbf{Z}_l & \mathbf{W}_l
\end{pmatrix}\begin{pmatrix}
A & B \\ B & C
\end{pmatrix}\begin{pmatrix}
\mathbf{Z}_l \\ \mathbf{W}_l
\end{pmatrix} \Big]
\\ & \le \mathbb{E}_{(x,v)\sim \mu}\Big[\frac{\mathbf{Z}_l K \mathbf{Z}_l}{\gamma^2} + \frac{(1-2\tau)^2}{4}|\mathbf{Z}_l|^2+ \Big( \frac{1+2\tau}{2}|\mathbf{Z}_l|+|\mathbf{Z}_l+\gamma^{-1}\mathbf{W}_l|\Big)^2\Big]
\\ & \le \max(\gamma^{-2}L_K+1)\mathbb{E}_{(x,v)\sim \mu}[|\mathbf{Z}_l|^2]+2 b_l\le 2(a_l+b_l)
\end{aligned}
\end{equation}
with $A,B,C$ given in \eqref{eq:rho} and $\tau$ given in \eqref{eq:tau}. 
Here, we used that $(1-2\tau)^2/4+(1+2\tau)^2/2\le 1$.
Note that in the following we drop the initial condition in the subscript of the expectation for simplicity.

We bound for all $i\in\mathbb{N}$, 
\begin{align}
 \mathbb{E}\Big[\int_0^h |Z_{ih+s}|^2 \rmd s\Big] & =\mathbb{E}\Big[\int_0^h \Big|\mathbf{Z}_{i}+\int_0^s W_{ih+r}\rmd r\Big|^2 \rmd s\Big]  \nonumber
 \\ & \le 2 h \hat{a}_i + 2 \mathbb{E}\Big[\int_0^h s \int_0^s |W_{ih+r}|^2 \rmd r \rmd s\Big] \nonumber
\\ &\le 2h \hat{a}_i + h^2 \mathbb{E}\Big[\int_0^h  |W_{ih+r}|^2 \rmd r\Big] \label{eq:intZ}
\end{align}
and
\begin{align*}
& \mathbb{E}\Big[\int_0^h |W_{ih+s}|^2 \rmd s\Big]  =\mathbb{E}\Big[\int_0^h\Big|\mathbf{W}_i+ \int_0^s (-\gamma W_{ih+r}- K Z_{ih+r}-\nabla G(X_{ih+r} )) \rmd r \Big|^2 \rmd s\Big]
\\ & \le 4 \int_0^h (2 \gamma^2 (\hat{a}_i+b_i))\rmd s + 4 \mathbb{E}\Big[ \int_0^h \Big|\int_0^s (\gamma W_{ih+r} + \gamma^2 Z_{ih+r} )\rmd r \Big|^2 \rmd s\Big]
\\ & \quad + 4 \mathbb{E}\Big[\int_0^h \Big|\int_0^s (\gamma^2-K) Z_{ih+r} \rmd r \Big|^2 \rmd s\Big]+4 \mathbb{E}\Big[\int_0^h \Big|\int_0^s \nabla G (X_{ih+r})\rmd r \Big|^2 \rmd s \Big]
\\ & \le 8 h\gamma^2 (\hat{a}_i+b_i)
\\ & \quad + 4\gamma^4 \mathbb{E}\Big[\int_0^h s\int_0^s \Big|\gamma^{-1} \mathbf{W}_i +\mathbf{Z}_i + \gamma^{-1}\int_0 ^r - KZ_{ih+u}-\nabla G(X_{ih+u})\rmd u\Big|^2 \rmd r \rmd s \Big]
\\ & \quad + 4 \mathbb{E}\Big[\int_0^h s\int_0^s \max(\gamma^4,L_K^2) |Z_{ih+r}|^2 \rmd r  \rmd s\Big]+4 \int_0^h s \int_0^s L_G^2  \max_{u \ge 0} \mathbb{E}[|X_{u}|^2] \rmd r  \rmd s
\\ & \le 8 h\gamma^2 (\hat{a}_i+b_i)+ 8\gamma^4 \frac{h^3}{3}b_i
\\ & \quad + 8\gamma^2 \mathbb{E}\Big[ \int_0^h s\int_0^s  r \int_0 ^r \Big|-KZ_{ih+u}-\nabla G(X_{ih+u})\Big|^2\rmd u \rmd r \rmd s \Big]
\\ & \quad + 4\max(\gamma^4,L_K^2) \int_0^h  \Big(2s^2 \hat{a}_i+s^3 \mathbb{E}\Big[\int_0^s |W_{ih+r}|^2\rmd r\Big]\Big) \rmd s  +\frac{4h^3}{3} L_G^2  \max_{u \ge 0} \mathbb{E}[|X_{u}|^2] 
\\ & \le 8 h\gamma^2 ((\hat{a}_i+b_i)+  \frac{\gamma^2 h^2}{3}b_i)+ 16 \gamma^2 L_K^2 \int_0^h s\int_0^s  r  \Big( 2 r \hat{a}_i + r^2 \mathbb{E}\Big[ \int_0 ^r  |W_{ih+u}|^2\rmd u\Big]\Big) \rmd r \rmd s
\\& \quad + 16 \gamma^2 \int_0^h s\int_0^s  r \int_0 ^r  L_G^2 \max_{v\ge 0}\mathbb{E}[|X_{v}|^2]\rmd u \rmd r \rmd s+ 8\max(\gamma^4,L_K^2)\frac{h^3 \hat{a}_i}{3}
\\ & \quad + \max(\gamma^4,L_K^2) h^4 \mathbb{E}\Big[\int_0^h |W_{ih+r}|^2 \rmd r \Big] +\frac{4h^3}{3} L_G^2  \max_{u \ge 0} \mathbb{E}[|X_{u}|^2] 
\\ & \le 8 h\gamma^2 (\hat{a}_i+b_i)+  \frac{8\gamma^4 h^3}{3}b_i+ \frac{32\gamma^2 L_K^2 h^5 a_i}{15} +\frac{16\gamma^2 L_K^2 h^6}{24} \mathbb{E}\Big[\int_0 ^h  |W_{ih+u}|^2\rmd u \Big]
\\& \quad + \frac{16 \gamma^2 h^5  L_G^2 \max_{u\ge 0}\mathbb{E}[|X_{u}|^2]}{15} + 8\max(\gamma^4,L_K^2)\frac{h^3 \hat{a}_i}{3}
\\ & \quad + \max(\gamma^4,L_K^2) h^4 \mathbb{E}\Big[\int_0^h |W_{ih+r}|^2 \rmd r\Big]  +\frac{4h^3}{3} L_G^2  \max_{u \ge 0} \mathbb{E}[|X_{u}|^2] ,
\end{align*}
where we used \eqref{eq:intZ} in the third and forth step. Hence, 
\begin{align}
\mathbb{E}\Big[\int_0^h  & |W_{ih+s}|^2 \rmd s \Big] \le \frac{1}{1-\max(\gamma^4,L_K^2) h^4- (2/3)\gamma^2 L_K^2 h^6 }\Big[\Big(8 h\gamma^2+  \frac{8\gamma^4 h^3}{3}\Big)b_i \nonumber
\\&  + \Big(8 h\gamma^2 + \frac{32\gamma^2 L_K^2 h^5 }{15} + 8\max(\gamma^4,L_K^2)\frac{h^3 }{3}\Big)\hat{a}_i \nonumber
\\ &  + \Big(\frac{16 \gamma^2 h^5  L_G^2}{15}    +\frac{4h^3}{3} L_G^2\Big)  \max_{u \ge 0} \mathbb{E}[|X_{u}|^2] \Big] \nonumber
\\ & \le \frac{96}{89}\Big[\Big(8 h\gamma^2+  \frac{2\gamma^2 h}{3}\Big)b_i  + \Big(8 h\gamma^2 + \frac{2\gamma^2 h }{15} + \frac{2\gamma^2 h }{3}\Big)\hat{a}_i + \frac{8 h^3  L_G^2}{5}   \max_{u \ge 0} \mathbb{E}[|X_{u}|^2] \Big] \nonumber
\\ & \le \frac{96}{89}\Big[ \frac{132\gamma^2 h }{15} (\hat{a}_i+b_i) + \frac{8 h^3  L_G^2}{5}   \max_{u \ge 0} \mathbb{E}[|X_{u}|^2] \Big], \label{eq:intW}
\end{align}
by applying \eqref{eq:condh} and using that by \eqref{eq:condh}, $\max(\gamma^4,L_K^2) h^4\le 1/16$ and $(2/3)\gamma^2 L_K^2 h^6\le 1/96$.
Further, it holds for $k\in\mathbb{N}$,
\begin{align*}
a_{k+1}&=\Big(\frac{L_K}{2\gamma^2}+1\Big)\mathbb{E}\Big[\Big|\mathbf{Z}_k+ \int_{0}^h W_{kh+r} \rmd r\Big|^2\Big]=\Big(\frac{L_K}{2\gamma^2}+1\Big)\mathbb{E}\Big[\Big|\mathbf{Z}_0+ \sum_{i=0}^k\int_{0}^h W_{ih+r} \rmd r\Big|^2\Big]
\\ & \le \Big(\frac{L_K}{2\gamma^2}+1\Big)(k+1)\sum_{i=1}^k h \mathbb{E}\Big[\int_0^h|W_{ih+r}|^2 \rmd r\Big]
\end{align*}
and 
\begin{align*}
b_{k+1}&=\mathbb{E}\Big[\gamma^{-2}\Big|\sum_{i=0}^k \int_0^h \Big(-KZ_{ih+r}- \nabla G(X_{ih+r})+\nabla G (\mathbf{X}_{i+1})\Big)\rmd r \Big|^2\Big]
\\ &\le 3 \gamma^{-2} (k+1)h \sum_{i=0}^k  \mathbb{E}\Big[\int_0^h ( |KZ_{ih+r}|^2 + | \nabla G(X_{ih+r})-\nabla G(X_{(i+1)h})|^2
\\ & \qquad + |\nabla G(X_{(i+1)h})-\nabla G(\mathbf{X}_{i+1})|^2)  \rmd r \Big]
\\ & \le 3 \gamma^{-2} (k+1)h \sum_{i=0}^k\Big(L_K^2 \Big(2 \hat{a}_i h+ h^2 \mathbb{E}\Big[\int_0^h |W_{ih+r}|^2 \rmd r\Big]\Big) 
\\ & \qquad + L_G^2 \mathbb{E}\Big[\int_0^h \Big|\int_r^h V_{ih+ s} \rmd s \Big|^2 \rmd r\Big]
 + h L_G^2\hat{a}_{i+1}   \Big)
\\ & \le  3 \gamma^{-2} (k+1)h \sum_{i=0}^k\Big(L_K^2 2 \hat{a}_i h+ L_K^2 h^2 \mathbb{E}\Big[\int_0^h |W_{ih+r}|^2 \rmd r\Big] + L_G^2 \frac{h^3}{3} \max_{u\ge 0}\mathbb{E}[|V_{u}|^2]
\\ & \qquad  + 2h L_G^2\hat{a}_{i}+ 2h^2L_G^2 \mathbb{E}\Big[\int_0^h|W_{ih+r}|^2\rmd r \Big]  \Big).
\end{align*}
Since $\mu$ is the invariant measure of the dynamics given by \eqref{eq:LD}, it holds
$\mathbb{E}[|X_u|^2]=\mathbb{E}[|X_0|^2]$ and $\mathbb{E}[|V_u|^2]=\mathbb{E}[|V_0|^2]$ for all $u\ge 0$.
Using this observation and \eqref{eq:intW} and combining the two estimates on $a_{k+1}$ and $b_{k+1}$, we obtain
\begin{align*}
& a_{k+1} +b_{k+1}  \le (k+1)\sum_{i=0}^k \Big( 6 \gamma^{-2}h^2 (L_K^2+L_G^2)  \hat{a}_i + 3 \gamma^{-2} L_G^2 \frac{h^4}{3} \mathbb{E}[|V_0|^2]
\\ &  \qquad +h\Big(\Big(\frac{L_K}{2\gamma^2}+1\Big)+3 \gamma^{-2} L_K^2 h^2+6 \gamma^{-2} h^2 L_G^2 \Big)  \mathbb{E}\Big[\int_0^h|W_{ih+r}|^2 \rmd r\Big]  \Big)
\\ & \le (k+1)\sum_{i=0}^k \Big( 6 \gamma^{-2}h^2 (L_K^2+L_G^2)  \hat{a}_i + \gamma^{-2} L_G^2 h^4 \mathbb{E}[|V_0|^2]
\\ & \qquad  +h\Big(\frac{L_K}{2\gamma^2}+1+3 \frac{  h^2(L_K^2+2L_G^2)}{\gamma^2}\Big)  \frac{96}{89}\Big[ \frac{132\gamma^2 h }{15} (\hat{a}_i+b_i) + \frac{8 h^3  L_G^2}{5}   \mathbb{E}[|X_0|^2] \Big]   \Big)
\\ & \le (k+1)\sum_{i=0}^k \Big(  \frac{6h^2 (L_K^2+L_G^2)}{\gamma^{2}} +\Big(\frac{L_K}{2\gamma^2}+1+3 \frac{h^2 (L_K^2 +2 L_G^2)}{\gamma^2} \Big)  \frac{96}{89}\frac{132\gamma^2 h^2 }{15}\Big) (\hat{a}_i+b_i) 
\\ & \qquad  +\Big(\frac{L_K}{2\gamma^2}+1+3 \gamma^{-2}h^2 (L_K^2 +2 L_G^2) \Big)  \frac{96}{89} \frac{8 h^4  L_G^2}{5}   \mathbb{E}[|X_0|^2]  +  \gamma^{-2} L_G^2 h^4 \mathbb{E}[|V_0|^2]  \Big).
\end{align*}
Since $\hat{a}_i\le a_i$, it holds
\begin{align*}
a_{k+1}+b_{k+1} 
 \le  (k+1)\sum_{i=0}^k h^2 \tilde{\lambda} (a_i+b_i) + h^4 (k+1)^2 C_1,
\end{align*}
with
\begin{align*} 
& \tilde{\lambda}= \Big(6 \gamma^{-2} (L_K^2+L_G^2) +\Big(\frac{L_K}{2\gamma^2}+1+3 \gamma^{-2}h^2 (L_K^2 +2 L_G^2) \Big)  \frac{96}{89}\frac{132\gamma^2  }{15}\Big),
\\ & \tilde{C}_1= \Big(\frac{L_K}{2\gamma^2}+1+3 \gamma^{-2}h^2 (L_K^2 +2 L_G^2) \Big)  \frac{96}{89} \frac{8  L_G^2}{5}   \mathbb{E}[|X_0|^2]  +  \gamma^{-2} L_G^2 \mathbb{E}[|V_0|^2]. 
\end{align*}
Using \eqref{eq:condh} and the fact that by \cite[Theorem 5.1]{BrLi1976}, $\mathbb{E}[|X_0|^2]=(d/\kappa)$ and $\mathbb{E}[|V_0|^2]=d$, we bound these two constants by
\begin{align*}
& \tilde{\lambda}\le \Big(6 \gamma^{-2} L_K^2+5L_K+ \frac{3}{2}\gamma^2 +\frac{17}{8}  \frac{96}{89}\frac{132\gamma^2  }{15}\Big)\le 6 \gamma^{-2} L_K^2 +5L_K+ 22 \gamma^{2}=:\lambda,
\\ & \tilde{C}_1\le \frac{L_K \gamma^{-2}d L_G^2}{\kappa} + \frac{17}{8}  \frac{96}{89} \frac{8  L_G^2}{5}   \frac{d}{\kappa}  +  \gamma^{-2} L_G^2 d\le  d L_G^2\Big(\frac{2L_K \gamma^{-2}}{\kappa}+ \frac{4}{\kappa}\Big) =: C_1,
\end{align*}
where we used \eqref{eq:condh} and $1\le L_K/\kappa$.
We fix $l\in\mathbb{N}$ such that $l = \lfloor \lambda^{-1/2} h^{-1} \rfloor $. Then, $k\in\mathbb{N}$ such that $(k+1)\le l$, it holds $ (k+1) h \le hl \le \lambda^{-1/2}$ and 
\begin{align*}
a_{k+1}+b_{k+1} \le \sum_{i=0}^k h \lambda^{1/2} (a_i+b_i) + h^2 C_1/\lambda
\end{align*}
and we observe that there exists a sequence $(c_k)_{k\in\mathbb{N}}$ satisfying $a_k+b_k\le c_k$ for $k\le h^{-1}\lambda^{-1/2}$, $c_1=h^2 C_1/\lambda$  and 
\begin{align*}
c_{k+1} = \sum_{i=1}^k h \sqrt{\lambda} c_i + h^2  C_1/\lambda= c_k+h \sqrt{\lambda} c_k= (1+h\sqrt{\lambda})^k c_1.
\end{align*}
Hence, for all $k\in \mathbb{N}$ such that $k+1\le l$
\begin{align*}
a_{k+1}+b_{k+1} &\le c_{k+1}\le  (1+h\sqrt{\lambda})^k c_1 \le e^{h k \sqrt{\lambda}} c_1\le e^{1} h^2 C_1/\lambda.
\end{align*}
Then, by \eqref{eq:W2_stracc} and \eqref{eq:W2_stracc2} we obtain
\begin{align*}
\mathcal{W}_{2,\rho}(\mu, \mu_h)&\le \Big(1+\frac{1}{clh}\Big)\mathcal{W}_{2,\rho}(\mu p_{lh}, \mu \pi_h^l)
\le \Big(1+\frac{1}{clh}\Big)\sqrt{2(a_l+b_l)}
\\ &  \le \Big(1+\frac{2\sqrt{\lambda}}{c}\Big) h \sqrt{2e} \frac{\sqrt{C_1}}{\sqrt{\lambda}}
 \le \Big(\frac{1}{\sqrt{\lambda}}+\frac{2}{c}\Big)   h \sqrt{2e}  \sqrt{d} L_G \gamma^{-1} \sqrt{\frac{4\gamma^2}{\kappa} + \frac{2L_K}{\kappa} }.
\end{align*}
Note that we used in the last step that by $l=\lfloor \lambda^{-1/2} h^{-1}\rfloor$ it holds $lh>2 \lambda^{-1/2}$.
We note that since $(\lambda)^{-1/2}\le (8c)^{-1/2}$, it holds
\begin{align*}
\mathcal{W}_{2,\rho}(\mu, \mu_h)&\le 8 c^{-1}   h \sqrt{d} L_G \gamma^{-1} \sqrt{\frac{2\gamma^2}{\kappa} + \frac{L_K}{\kappa} },
\end{align*}
which concludes the proof.
\end{proof}

\begin{proof}[Proof of Theorem~\ref{thm:complex}]
Applying the triangle inequality and Theorem~\ref{thm:conv} and Theorem~\ref{thm:strongacc}, we obtain for all $k\in\mathbb{N}$,
\begin{align*}
\mathcal{W}_{2,\rho}(\mu, \nu \pi_h^k)&\le \mathcal{W}_{2,\rho}(\mu, \mu_h)+\mathcal{W}_{2,\rho}(\mu_h, \nu \pi_h^k)
\\ & \le h\Big(8 c^{-1}    \sqrt{d} L_G \gamma^{-1} \sqrt{\frac{2\gamma^2}{\kappa} + \frac{L_K}{\kappa} }\Big)+e^{-chk} \mathcal{W}_{2,\rho}(\mu_h, \nu ).
\end{align*}
The bound in $L^2$ Wasserstein distance with respect to the Euclidean distance is obtained by using the equivalence of the distance $\rho$ and the Euclidean distance stated in \eqref{eq:equivdist}.
\end{proof}

\begin{proof}[Proof of Theorem~\ref{thm:stracc_PGP}]
Fix $h$ satisfying \eqref{eq:condh2}. Denote by $(p_t)_{t\ge0}$ the transition function of the exact continuous dynamics given by \eqref{eq:LD} and by $\tilde{\pi}_h$ the transition kernel of the Markov chain corresponding to the $\mathcal{PGP}$-splitting scheme \eqref{eq:PGP}. the measures $\mu$ and $\tilde{\mu}_h$ denote the invariant probability measures of the processes corresponding to \eqref{eq:LD} and \eqref{eq:PGP}, respectively. By triangle inequality and Theorem~\ref{thm:conv_PGP2}, it holds for $l\in\mathbb{N}$
\begin{align*}
\mathcal{W}_{2,\rho}(\mu,\tilde{\mu}_h)&\le \mathcal{W}_{2,\rho}(\mu p_{lh},\mu \tilde{\pi}_h^l)+\mathcal{W}_{2,\rho}(\mu \tilde{\pi}_h^l,\tilde{\mu}_h\tilde{\pi}_h^l)
\\ & \le \mathcal{W}_{2,\rho}(\mu p_{lh},\mu\tilde{\pi}_h^l)+e^{-chl}\mathcal{W}_{2,\rho}(\mu ,\tilde{\mu}_h).
\end{align*}
Hence, for $l\ge (ch)^{-1}$
\begin{align} \label{eq:W2_stracc_triangle}
\mathcal{W}_{2,\rho}(\mu,\tilde{\mu}_h)\le (1- e^{-chl})^{-1} \mathcal{W}_{2,\rho}(\mu p_{lh},\mu \tilde{\pi}_h^l)\le (1+(chl)^{-1}) \mathcal{W}_{2,\rho}(\mu p_{lh},\mu \tilde{\pi}_h^l).
\end{align}
We bound the Wasserstein distance on the right hand side by considering a synchronous coupling of the exact Langevin dynamics $(X_t, V_t)_{t\ge 0}$ given by \eqref{eq:LD} and with initial condition $(X_0, V_0)=(x,v)\sim \mu$ and the Markov chain $(\tilde{\mathbf{X}}_k,\tilde{\mathbf{V}}_k)_{k\in\mathbb{N}}$ given by \eqref{eq:PGP} with initial condition $(\tilde{\mathbf{X}}_0,\tilde{\mathbf{V}}_0)=(x,v)\sim \mu$. We represent $(\tilde{\mathbf{X}}_k,\tilde{\mathbf{V}}_k)_{k\in\mathbb{N}}$ by the following recursive scheme where the splitting steps are successively performed:
\begin{align*}
\mathcal{P}_{1/2}:\qquad  &(\tilde{\mathbf{X}}_k^P,\tilde{\mathbf{V}}_k^P)=(\tilde{\mathbf{X}}_k,\tilde{\mathbf{V}}_k-(h/2)\nabla G(\tilde{\mathbf{X}}_k)),
\\ \mathcal{G}:\qquad& (\tilde{\mathbf{X}}_{k+1}^G,\tilde{\mathbf{V}}_{k+1}^G)=(\tilde{X}_{(k+1)h}, \tilde{V}_{(k+1)h})
\\\mathcal{P}_{1/2}:\qquad & (\tilde{\mathbf{X}}_{k+1},\tilde{\mathbf{V}}_{k+1})=(\tilde{\mathbf{X}}_{k+1}^G,\tilde{\mathbf{V}}_{k+1}^G-(h/2)\nabla G(\tilde{\mathbf{X}}_{k+1}^G)),
\end{align*} 
where
\begin{align*}
&\begin{cases}
\rmd \tilde{X}_s=\tilde{V}_s \rmd s
\\ \rmd \tilde{V}_s=(-\gamma \tilde{V}_s-K\tilde{X}_s) \rmd s+ \sqrt{2\gamma}\rmd B_s
\end{cases}
\text{with } \tilde{X}_{kh}=\tilde{\mathbf{X}}_k^P, \qquad \tilde{V}_{kh}=\tilde{\mathbf{V}}_k^P.
\end{align*}

For the synchronous coupling, we consider the same Brownian motion for the two processes given by  \eqref{eq:LD} and \eqref{eq:dynPGP}, respectively, on a joint probability space. Then, it holds for $(\mathbf{Z}_k,\mathbf{W}_k)=(X_{kh}-\tilde{\mathbf{X}}_{k},V_{kh}-\tilde{\mathbf{V}}_{k})$
\begin{align}\label{eq:dynPGP}
\begin{cases}
\mathbf{Z}_{k+1}&=\mathbf{Z}_k+ \int_0^h W_{kh+r} \rmd r
\\  \mathbf{W}_{k+1}&=\mathbf{W}_k- \int_0^h (\gamma W_{kh+r}+KZ_{kh+r}+\nabla G(X_{kh+r})
\\ & -\frac{1}{2}\nabla G(\tilde{\mathbf{X}}_{k})-\frac{1}{2}\nabla G(\tilde{\mathbf{X}}_{k+1})) \rmd r,
\end{cases}
\end{align}
where $(Z_u,W_u)_{u\ge 0}=(X_{u}-\tilde{X}_u,V_{u}-\tilde{V}_u)_{u\ge 0}$ and $(\mathbf{Z}_0,\mathbf{W}_0)=(0,0)$.
Set for $k\in\mathbb{N}$
\begin{align*}
a_{k}=\Big(\frac{L_K}{2\gamma^2}+1\Big)\hat{a}_k=\Big(\frac{L_K}{2\gamma^2}+1\Big) \mathbb{E}[|\mathbf{Z}_k|^2], \qquad \text{and} \qquad 
b_k=\mathbb{E}[|\mathbf{Z}_k+\gamma^{-1} \mathbf{W}_k|^2].
\end{align*}
As in the proof of Theorem~\ref{thm:strongacc} the equation \eqref{eq:W2_stracc2} holds and by \eqref{eq:W2_stracc_triangle}
\begin{align} \label{eq:W2_stracc3}
\mathcal{W}_{2,\rho}(\mu ,\tilde{\mu}_h)\le (1+(chl)^{-1}) \sqrt{2(a_l+b_l)}.
\end{align}
Therefore, we provide next bounds for $a_k$ and $b_k$, $k\in\mathbb{N}$. By \eqref{eq:dynPGP}, it holds
\begin{align*}
a_{k+1}=\Big(\frac{L_K}{2\gamma^2}+1\Big) \mathbb{E}\Big[\Big|\mathbf{Z}_k+ \int_0^h W_{kh+s} \rmd s \Big|^2\Big]=\Big(\frac{L_K}{2\gamma^2}+1\Big) \mathbb{E}\Big[\Big|\sum_{i=0}^k \int_0^h W_{ih+s}\rmd s\Big|^2\Big],
\end{align*}
since $\mathbf{Z}_0=0$.
We write the integral by using multiple times \eqref{eq:dynPGP}
\begin{align*}
\int_0^h W_{ih+s}\rmd s& = \int_0^h \Big(\mathbf{W}_i+ \frac{h}{2}\nabla G(\tilde{\mathbf{X}}_i)-\int_0^s (\gamma W_{ih+r}+ K Z_{ih+r}+ \nabla G(X_{ih+r}))\rmd r\Big) \rmd s
\\ &= h \mathbf{W}_i+ \int_0^h \int_0^s (-\gamma W_{ih+r}- K Z_{ih+r}- \nabla G(X_{ih+r})+\nabla G(\tilde{\mathbf{X}}_i))\rmd r \rmd s
\\ & = h \gamma ( \mathbf{Z}_i+\gamma^{-1}  \mathbf{W}_i)- h\gamma \mathbf{Z}_i-\int_0^h \int_0^s \gamma W_{ih+r}\rmd r\rmd s
\\ & \qquad - K \int_0^h\int_0^s \int_0^r W_{ih+u} \rmd u \rmd r \rmd s 
 - K\frac{h^2}{2}\mathbf{Z}_i
 \\ & \qquad + \int_0^h \int_0^s (-\nabla G(X_{ih+r})+\nabla G(X_{ih})-\nabla G(X_{ih})+\nabla G(\tilde{\mathbf{X}}_i))\rmd r \rmd s.
\end{align*}
Hence, using Young's inequality (i.e., $(\sum_{i=1}^6 a_i)^2\le 6 \sum_{i=1}^6 a_i^6$), we obtain
\begin{align*}
\hat{a}_{k+1}=\mathbb{E}\Big[\Big|\sum_{i=0}^k \int_0^h W_{ih+s}\rmd s\Big|^2\Big]\le (k+1)\sum_{i=0}^k\mathbb{E}\Big[\Big|\int_0^h W_{ih+s}\rmd s\Big|^2\Big] \le \sum_{i=0}^k 6 \sum_{n=1}^6 I_{i,n}, 
\end{align*}
where
\begin{align*}
I_{i,1}&= \mathbb{E}\Big[\Big| h \gamma ( \mathbf{Z}_i+\gamma^{-1}  \mathbf{W}_i)\Big|^2\Big]\le h^2\gamma^2 b_i
\\  I_{i,2}&= \mathbb{E}\Big[\Big|\Big(h \gamma \mathbf{Z}_i+ K\frac{h^2}{2}\mathbf{Z}_i\Big)\Big|^2\Big]\le \Big(h\gamma+ L_K\frac{h^2}{2}\Big)^2 \hat{a}_i
\\  I_{i,3}&=\mathbb{E}\Big[\Big| \int_0^h\int_0^s (-\gamma W_{ih+r}) \rmd r \rmd s\Big|^2\Big]\le \gamma^2\mathbb{E}\Big[ h\int_0^hs\int_0^s |W_{ih+r}|^2 \rmd r \rmd s\Big]
\\ & \le  \gamma^2\frac{h^3}{2}\mathbb{E}\Big[ \int_0^h |W_{ih+r}|^2 \rmd r\Big]
\\  I_{i,4}&=\mathbb{E}\Big[\Big| K\int_0^h\int_0^s\int_0^r W_{ih+u} \rmd u \rmd r \rmd s\Big|^2\Big]
 \le L_K^2\mathbb{E}\Big[ h\int_0^hs\int_0^s r\int_0^r |W_{ih+u}|^2 \rmd u \rmd r \rmd s\Big]
\\ & \le  \frac{L_K^2h^5}{8}\mathbb{E}\Big[ \int_0^h |W_{ih+u}|^2 \rmd u\Big]
\\  I_{i,5}&=\mathbb{E}\Big[\Big| \int_0^h \int_0^s \nabla G(X_{ih+r})-\nabla G(X_{ih}))\rmd r \rmd s\Big|^2\Big]
\\ & \le  \mathbb{E}\Big[ h\int_0^h s\int_0^s L_G^2|\int_0^r V_{ih+u} \rmd u|^2 \rmd r \rmd s\Big]
\\ & \le L_G^2  h\int_0^h s\int_0^sr\int_0^r \max_{v\ge 0}\mathbb{E}[|V_{v}|^2] \rmd u \rmd r \rmd s \le  L_G^2 \frac{h^6}{15}  \max_{v\ge 0}\mathbb{E}[|V_{v}|^2]
\\ I_{i,6}&=\mathbb{E}\Big[\Big| \int_0^h \int_0^s \nabla G(X_{ih})-\nabla G(\tilde{\mathbf{X}}_i))\rmd r \rmd s\Big|^2\Big]\le  \frac{h^4}{4}L_G^2  \hat{a}_i.
\end{align*}
In particular, using that by \eqref{eq:condh2}, $L_K^2h^2\le \gamma^2/(16)$ and that $b_0=\hat{a}_0=0$
\begin{align}\label{eq:boundhata}
\hat{a}_{k+1}& =\mathbb{E}\Big[\Big|\sum_{i=0}^k  \int_0^h  W_{ih+s}\rmd s\Big|^2\Big] \nonumber
\\ & \le (k+1) \sum_{i=1}^k\Big(6 h^2 \gamma^2 b_i+6\Big(\Big(h\gamma+ L_K\frac{h^2}{2}\Big)^2+\frac{h^4}{4}L_G^2\Big) \hat{a}_i\Big) \nonumber
\\ & + 3(k+1)\frac{65\gamma^2 h^3}{64}\sum_{i=0}^k \mathbb{E}\Big[\int_0^h |W_{ih+r}|^2 \rmd r\Big]+ 2\frac{(k+1)^2 L_G^2 h^6}{5}  \max_{v\ge 0}\mathbb{E}[|V_{v}|^2].
\end{align}
Similar to the above calculation we bound the second last term by 
\begin{align*}
\mathbb{E}\Big[\int_0^h |W_{ih+s}|^2 \rmd s\Big]
& = \mathbb{E}\Big[\int_0^h\Big|\gamma( \mathbf{Z}_i+\gamma^{-1}  \mathbf{W}_i)- \gamma \mathbf{Z}_i+ \int_0^s (-\gamma W_{ih+r})\rmd r - Ks\mathbf{Z}_i
 \\ & - K\int_0^s \int_0^r W_{ih+u} \rmd u \rmd r  + \int_0^s (-\nabla G(X_{ih+r}))\rmd r+\frac{h}{2}\nabla G(X_{ih})
 \\ & +\frac{h}{2}
(-\nabla G(X_{ih})+\nabla G(\tilde{\mathbf{X}}_i))\Big|^2\rmd s \Big]\le 6 \sum_{n=1}^6 J_{i,n},
\end{align*}
where $J_{i,n}$ are similarly bounded as $I_{i,n}$ by
\begin{align*}
J_{i,1} & = \mathbb{E}\Big[\int_0^h\Big|\gamma( \mathbf{Z}_i+\gamma^{-1}  \mathbf{W}_i)\Big|^2\rmd s\Big]\le \gamma^2 h b_i
\\ J_{i,2} & =\mathbb{E}\Big[\int_0^h\Big|\gamma \mathbf{Z}_i+ Ks \mathbf{Z}_i\Big|^2\rmd s\Big]\le (\gamma+L_K h)^2 h \hat{a}_i
\\ J_{i,3} & =\mathbb{E}\Big[\int_0^h\Big|\int_0^s (-\gamma W_{ih+r})\rmd r\Big|^2\rmd s\Big]\le \frac{\gamma^2 h^2}{2}\mathbb{E}\Big[\int_0^h |W_{ih+r}|^2 \rmd r\Big]
\\ J_{i,4} & =\mathbb{E}\Big[\int_0^h\Big|K\int_0^s \int_0^r W_{ih+u} \rmd u \rmd r\Big|^2\rmd s\Big]\le \frac{L_K^2 h^4}{8}\mathbb{E}\Big[\int_0^h |W_{ih+u}|^2 \rmd u\Big]
\\ J_{i,5} & =\mathbb{E}\Big[\int_0^h\Big|\int_0^s \nabla G(X_{ih+r})\rmd r-\frac{h}{2}\nabla G(X_{ih})\Big|^2\rmd s\Big] 
\\ & \le 2\mathbb{E}[h L_G^2(\frac{h}{2})^2|X_{ih}|^2]+ 2L_G^2\int_0^h s \int_0^s\mathbb{E}[|X_{ih+r}|^2]\rmd r \rmd s
\\ & \le \frac{h^3 L_G^2}{2}\mathbb{E}[|X_{ih}|^2]+ \frac{2h^3 L_G^2}{3}\max_{s\ge 0}\mathbb{E}[|X_s|^2] \le \frac{7}{6}h^3 L_G^2 \max_{s\ge 0}\mathbb{E}[|X_s|^2]
\\ J_{i,6} & =\mathbb{E}\Big[\int_0^h\Big|\frac{h}{2}(
\nabla G(X_{ih})-\nabla G(\tilde{\mathbf{X}}_i))\Big|^2\rmd s\Big]\le L_G^2 \frac{h^3}{4}\hat{a}_i.
\end{align*}
Subtracting the bound of $J_{i,3}$ and $J_{i,4}$ and dividing both sides by $1-3\gamma^2h^2 -(3/4)L_K^2h^2$, we obtain
\begin{align} \label{eq:intW2}
\mathbb{E}\Big[\int_0^h |W_{ih+s}|^2 & \rmd s\Big]\le \frac{6}{1-3\gamma^2h^2 -(3/4)L_K^2h^2}\sum_{n\in\{1,2,5,6\}} J_{i,n}\le \frac{6144}{829}\sum_{n\in\{1,2,5,6\}} J_{i,n} \nonumber
\\ & \le \frac{6144}{829}\Big(\gamma^2 h b_i+ (\gamma+L_K h)^2 h \hat{a}_i+\frac{7}{6}h^3 L_G^2 \max_{s\ge 0}\mathbb{E}[|X_s|^2]+ L_G^2 \frac{h^3}{4}\hat{a}_i\Big),
\end{align}
where the second step holds by \eqref{eq:condh2}, since $3\gamma^2h^2+(3/4)L_K^2h^4\le 195/1024$.
Inserting it back into \eqref{eq:boundhata}, yields
\begin{align*}
\hat{a}_{k+1}&\le 6 (k+1) \sum_{i=1}^k \Big( h^2\gamma^2b_i+\Big(\Big(h\gamma+ L_K\frac{h^2}{2}\Big)^2+\frac{h^4}{4}L_G^2\Big) \hat{a}_i\Big)
\\ & \quad 
 + 3(k+1)\gamma^2 h^3\sum_{i=0}^k \frac{6240}{829}\Big(\gamma^2 h b_i+ (\gamma+L_K h)^2 h \hat{a}_i
 \\ & \quad +\frac{7}{6}h^3 L_G^2 \max_{s\ge 0}\mathbb{E}[|X_s|^2]+ L_G^2 \frac{h^3}{4}\hat{a}_i\Big)+ 2\frac{(k+1)^2 L_G^2 h^6}{5}  \max_{v\ge 0}\mathbb{E}[|V_{v}|^2].
\end{align*}
Bounding $b_k$ we observe
\begin{align*}
 b_{k+1}
& = \mathbb{E}\Big[\Big|\mathbf{Z}_k+\gamma^{-1} \mathbf{W}_k+\frac{1}{\gamma}\int_0^h (-KZ_{kh+r}-\nabla G(X_{kh+r}))\rmd r
\\ &  +\frac{h}{2\gamma}(\nabla G(\tilde{X}_k)+\nabla G(\tilde{X}_{k+1}))\Big|^2\Big]
\\ & \le 2\gamma^{-2} \mathbb{E}\Big[\Big|\sum_{i=0}^k\int_0^h (-KZ_{ih+r}-\nabla G(X_{ih+r}))\rmd r+\frac{h}{2}(\nabla G({X}_i)+\nabla G({X}_{i+1}))\Big|^2\Big]
\\ & \quad + 2\gamma^{-2} \mathbb{E}\Big[\Big|\sum_{i=0}^k \frac{h}{2}(\nabla G(\tilde{X}_i)+\nabla G(\tilde{X}_{i+1})-\nabla G({X}_i)-\nabla G({X}_{i+1}))\Big|^2\Big]
\\ & \le 4\gamma^{-2} \mathbb{E}\Big[\Big|\sum_{i=0}^k K \int_0^h \Big(\mathbf{Z}_i +\int_0^r W_{ih+u}\rmd u\Big)\rmd r\Big|^2\Big]
\\ & \quad+4\gamma^{-2} \mathbb{E}\Big[\Big|\sum_{i=0}^k \Big(-\int_0^h\nabla G(X_{ih+r}))\rmd r+\frac{h}{2}(\nabla G({X}_i)+\nabla G({X}_{i+1}))\Big)\Big|^2\Big]
\\ & \quad+ 4\gamma^{-2} \mathbb{E}\Big[\Big|\sum_{i=0}^k \frac{h}{2}(\nabla G(\tilde{X}_i)-\nabla G({X}_i))\Big|^2\Big]
\\ & \quad + 4\gamma^{-2} \mathbb{E}\Big[\Big|\sum_{i=0}^k \frac{h}{2}(\nabla G(\tilde{X}_{i+1})-\nabla G({X}_{i+1}))\Big|^2\Big]
 = \sum_{n=1}^4 \tilde{I}_n.
\end{align*}
In the same spirit as $I_{i,2}$ and $I_{i,3}$ we bound
\begin{align} \label{eq:tildeI1}
\tilde{I}_1\le 8 \gamma^{-2} k \sum_{i=0}^k L_K^2 h^2 \hat{a}_i+ 8\gamma^{-2} (k+1) \sum_{i=0}^k L_K^2 \frac{h^3}{2}\mathbb{E}\Big[\int_0^h |W_{ih+r}|^2 \rmd r\Big].
\end{align}
Similar as $I_{i,6}$, we bound
\begin{align} \label{eq:tildeI34}
\tilde{I}_3+\tilde{I}_4\le \gamma^{-2} (k+1) \sum_{i=0}^k h^2 L_G^2 (\hat{a}_i+\hat{a}_{i+1}).
\end{align}
Using Lemma~\ref{lem:trapez}, we bound $\tilde{I}_2$ by 
\begin{align*}
\tilde{I}_2&\le 4\gamma^{-2} \mathbb{E}\Big[\Big|\sum_{i=0}^k \Big\{ \frac{h}{2}\int_0^h \int_0^r \nabla^2 G(X_{ih+u})\sqrt{2\gamma}\rmd B_{ih+u} \rmd r
\\&  -\int_0^h \int_0^s \int_0^r \nabla^2 G(X_{ih+u})\sqrt{2\gamma}\rmd B_{ih+u} \rmd r \rmd s
\\ & + \frac{h}{2}\int_0^h \int_0^s \nabla^3 G(X_{ih+u})[V_{ih+u},V_{ih+u}] \rmd u \rmd s
\\ &  - \int_0^h \int_0^s \int_0^r \nabla^3 G(X_{ih+u})[V_{ih+u},V_{ih+u}] \rmd u \rmd r \rmd s 
\\ & + \frac{h}{2}\int_0^h \int_0^r \nabla^2 G(X_{ih+u})(-\gamma V_{ih+u}-KX_{ih+u}-\nabla G(X_{ih+u})) \rmd u \rmd r
\\ &  -\int_0^h \int_0^s \int_0^r  \nabla^2 G(X_{ih+u})(-\gamma V_{ih+u}-KX_{ih+u}-\nabla G(X_{ih+u})) \rmd u \rmd r \rmd s\Big\}\Big|^2\Big].
\end{align*}
Hence
\begin{align*}
\tilde{I}_2 \le 4\gamma^{-2} 6 \sum_{j=1}^6 \tilde{J}_j,
\end{align*}
where $(\tilde{J}_j)_{j=1,\ldots,6}$ are bounded in the following way: For $\tilde{J}_1$ we obtain by Jensen's inequality, Ito's isometry and since the Brownian motions are independent over disjoints intervals
\begin{align*}
\tilde{J}_1&=\mathbb{E}\Big[\Big|\sum_{i=0}^k \frac{h}{2}\int_0^h \int_0^r \nabla^2 G(X_{ih+u})\sqrt{2\gamma}\rmd B_{ih+u} \rmd r\Big|^2 \Big]
\\ &\le \frac{h^3}{4} \int_0^h \mathbb{E}\Big[ \sum_{i=0}^k \Big|\int_0^r \nabla^2 G(X_{ih+u})\sqrt{2\gamma}\rmd B_{ih+u}\Big|^2\Big]\rmd r
\\ & \le \frac{h^3}{4} \int_0^h \mathbb{E}\Big[ \sum_{i=0}^k \int_0^r 2\gamma\|\nabla^2 G(X_{ih+u})\|_F^2 \rmd u \Big]\rmd r
 \le \frac{h^3}{4} \int_0^h  \sum_{i=0}^k \int_0^r (2\gamma d L_G^2)\rmd u \rmd r
 \\ & \le \frac{h^5}{8}(k+1) 2\gamma L_G^2 d,
\end{align*}
where we used Assumption~\ref{ass:U} in the second last step.
Analogously we bound $\tilde{J}_2$ by
\begin{align*}
\tilde{J}_2&=\mathbb{E}\Big[\Big|\sum_{i=0}^k \int_0^h \int_0^s \int_0^r \nabla^2 G(X_{ih+u})\sqrt{2\gamma}\rmd B_{ih+u} \rmd r \rmd s\Big|^2 \Big]
\\ &\le h \int_0^h s \int_0^s \mathbb{E}\Big[ \sum_{i=0}^k \int_0^r 2\gamma \|\nabla^2 G(X_{ih+u})\|_F^2 \rmd u \Big]\rmd r \rmd s \le \frac{h^5}{8}(k+1) 2\gamma L_G^2 d.
\end{align*}
By Assumption~\ref{ass_GHess} and Jensen's inequality, we bound $\tilde{J}_3$ and $\tilde{J}_4$ by
\begin{align} 
\tilde{J}_3&=\mathbb{E}\Big[\Big|\sum_{i=0}^k \frac{h}{2}\int_0^h \int_0^r \nabla^3 G(X_{ih+u})[V_{ih+u},V_{ih+u}] \rmd u \rmd r\Big|^2 \Big] \nonumber
\\ & \le (k+1) \sum_{i=0}^k \frac{h^2}{4} h\int_0^h r \int_0^r \mathbb{E}[|\nabla^3 G(X_{ih+u}) \nonumber [V_{ih+u},V_{ih+u}]|^2]\rmd u \rmd r 
\\ & \le (k+1)^2 \frac{h^6}{12} \max_{s\ge 0}\mathbb{E}[L_H^2|V_{s}|^4],  \label{eq:badterm1}
\\ \tilde{J}_4&=\mathbb{E}\Big[\Big|\sum_{i=0}^k \int_0^h \int_0^s \int_0^r \nabla^3 G(X_{ih+u})[V_{ih+u},V_{ih+u}] \rmd u \rmd r \rmd s \Big|^2 \Big] \nonumber
\\ & \le (k+1)^2 \frac{h^6}{15} \max_{s\ge 0}\mathbb{E}[L_H^2|V_{s}|^4]. \label{eq:badterm2}
\end{align}
For $\tilde{J}_5$ and $\tilde{J}_6$ we observe
\begin{align*}
\tilde{J}_5&=\mathbb{E}\Big[\Big|\sum_{i=0}^k \frac{h}{2}\int_0^h \int_0^r \nabla^2 G(X_{ih+u})(-\gamma V_{ih+u}-KX_{ih+u}-\nabla G(X_{ih+u})) \rmd u \rmd r  \Big|^2  \Big]
\\ & \le (k+1) \sum_{i=0}^k \frac{h^2}{4} h\int_0^h r \int_0^r \mathbb{E}[|\nabla^2 G(X_{ih+u})(-\gamma V_{ih+u}-KX_{ih+u}
\\ & -\nabla G(X_{ih+u}))|^2]\rmd u \rmd r 
\\ & \le (k+1)^2 \frac{h^6}{12} \max_{s\ge 0}\mathbb{E}[L_G^2(2\gamma^2|V_{s}|^2+2|\nabla U(X_s)|^2)],
\\ \tilde{J}_6&=\mathbb{E}\Big[\Big|\sum_{i=0}^k \int_0^h \int_0^s \int_0^r  \nabla^2 G(X_{ih+u})(-\gamma V_{ih+u}-KX_{ih+u}-\nabla G(X_{ih+u})) \rmd u \rmd r \rmd s \Big|^2 \Big]
\\ & \le (k+1)^2 \frac{h^6}{15} \max_{s\ge 0}\mathbb{E}[L_G^2(2\gamma^2|V_{s}|^2+2|\nabla U(X_s)|^2)].
\end{align*}
Inserting the bounds back into $\tilde{I}_2$ yields
\begin{align*}
\tilde{I}_2\le 24 \gamma^{-2} \sum_{j=1}^6 \tilde{J}_j\le 24\gamma^{-2} &\Big(\frac{h^5(k+1) }{2}L_G^2\gamma d + \frac{3(k+1)^2h^6}{20}\max_{s\ge 0} \mathbb{E}[L_H^2|V_s|^4]
\\ & + \frac{3(k+1)^2h^6}{20}\max_{s\ge 0} \mathbb{E}[L_G^2(2\gamma^2|V_s|^2+2 |\nabla U(X_s)|^2)]\Big).
\end{align*}
Combining this bound with \eqref{eq:tildeI1} and \eqref{eq:tildeI34} and using \eqref{eq:intW2} yields for $b_{k+1}$
\begin{align*}
b_{k+1}&\le 8 \gamma^{-2} k \sum_{i=0}^k L_K^2 h^2 \hat{a}_i+ 8\gamma^{-2} (k+1) \sum_{i=0}^k \frac{h^3L_K^2}{2}\mathbb{E}\Big[\int_0^h |W_{ih+r}|^2 \rmd r\Big]
\\& + \gamma^{-2} (k+1) \sum_{i=0}^k h^2 L_G^2 (\hat{a}_i+\hat{a}_{i+1})
 + 24\gamma^{-2} \Big( \frac{3(k+1)^2h^6}{20}\max_{s\ge 0} \mathbb{E}[L_H^2|V_s|^4]
\\ & +\frac{h^5(k+1) }{2}L_G^2\gamma d + \frac{3(k+1)^2h^6}{20}\max_{s\ge 0} \mathbb{E}[L_G^2(2\gamma^2|V_s|^2+2 |\nabla U(X_s)|^2)]\Big)
\\ &\le 8 \gamma^{-2} k \sum_{i=0}^k L_K^2 h^2 \hat{a}_i+ 8\gamma^{-2} (k+1) \sum_{i=0}^k \frac{h^3 L_K^2}{2}\frac{6144}{829}
\Big(\gamma^2 h b_i+ (\gamma+L_K h)^2 h \hat{a}_i\Big)
\\ & +8\gamma^{-2} (k+1) \sum_{i=0}^k \frac{h^3 L_K^2}{2}\frac{6144}{829}\Big(\frac{7}{6}h^3 L_G^2 \max_{s\ge 0}\mathbb{E}[|X_s|^2]+ L_G^2 \frac{h^3}{4}\hat{a}_i\Big)
\\& + \gamma^{-2} (k+1) \sum_{i=0}^k h^2 L_G^2 (\hat{a}_i+\hat{a}_{i+1})
 + 24\gamma^{-2} \Big(\frac{3(k+1)^2h^6}{20}\max_{s\ge 0} \mathbb{E}[L_H^2|V_s|^4]
\\ & +\frac{h^5(k+1) }{2}L_G^2\gamma d +  \frac{3(k+1)^2h^6}{20}\max_{s\ge 0} \mathbb{E}[L_G^2(2\gamma^2|V_s|^2+2 |\nabla U(X_s)|^2)]\Big).
\end{align*}
We recall that $\max_{s\ge 0}\mathbb{E}[|X_s|^2]=\mathbb{E}[|X_0|^2]\le d/\kappa$, $\max_{s\ge 0}\mathbb{E}[|V_s|^2]=\mathbb{E}[|V_0|^2]\le d$ and \\ $\max_{s\ge 0}\mathbb{E}[|V_s|^4]=\mathbb{E}[|V_0|^4]\le d^2+2d$. By \cite[Lemma A.3]{DePaBoDo2021}, $\max_{s\ge 0}\mathbb{E}[|\nabla U(X_s)|^2]=\mathbb{E}[|\nabla U(X_0)|^2]\le (L_G+L_K)d$ .
Inserting these bounds in the formulas of $\hat{a}_{k+1}$ and $b_{k+1}$ gives
\begin{align*}
\hat{a}_{k+1}&\le (k+1)\sum_{i=1}^k 6\Big(\Big(h\gamma+ L_K\frac{h^2}{2}\Big)^2+\frac{h^4}{4}L_G^2+\frac{3120}{829} \gamma^2h^4 \Big((\gamma+L_K h)^2 + L_G^2 \frac{h^2}{4}  \Big)\Big) \hat{a}_i
\\ & +(k+1)\sum_{i=1}^k 6\Big( h^2\gamma^2+ \frac{3120}{829}\gamma^4 h^4 \Big)b_i+ \mathbf{M}_1,
\end{align*}
and
\begin{align*}
b_{k+1}&\le   (k+1) \sum_{i=0}^k \Big(8 \gamma^{-2} L_K^2 h^2 + 4\gamma^{-2}h^4 \frac{6144}{829} L_K^2 \Big( (\gamma+L_K h)^2+ \frac{h^2 L_G^2}{4}\Big)  \Big) \hat{a}_i
\\ & +(k+1) \sum_{i=0}^k  4 h^4 \frac{6144}{829} L_K^2 b_i+ \gamma^{-2} (k+1) \sum_{i=0}^k h^2 L_G^2 (\hat{a}_i+\hat{a}_{i+1}) + \mathbf{M}_2
\end{align*}
with 
\begin{align*}
\mathbf{M}_1&=(k+1)^2 L_G^2 h^6\Big(\frac{2}{5}+ \frac{21840}{829}\frac{\gamma^2 }{\kappa}\Big) d
\\  \mathbf{M}_2&=(k+1)^2 \gamma^{-2}h^6 L_K^2 \frac{28672}{829} L_G^2 \frac{d}{\kappa}+  12\gamma^{-1}h^5(k+1)L_G^2 d 
\\ &  + \gamma^{-2}\frac{18(k+1)^2h^6}{5} L_H(d^2+2d)+ \gamma^{-2}\frac{18(k+1)^2h^6}{5}L_G^2(2\gamma^2d+2(L_K+L_G)d).
\end{align*}
Using \eqref{eq:condh}, we obtain for the prefactors of $\hat{a}_i$ and $b_i$ in the bounds of $\hat{a}_{k+1}$ and $b_{k+1}$
\begin{align*}
& 6\Big(\Big(h\gamma+ L_K\frac{h^2}{2}\Big)^2+\frac{h^4}{4}L_G^2+\frac{3120}{829} \gamma^2h^4 \Big((\gamma+L_K h)^2 + L_G^2 \frac{h^2}{4}  \Big)\Big)
\\ & \qquad \le 6 \Big(\Big(\frac{9}{8}\Big)^2+\frac{1}{256}+\frac{3120}{829} \frac{1}{16}\Big(\Big(\frac{5}{4}\Big)^2+\frac{1}{256}\Big)\Big)(h\gamma)^2 \le 10 h^2\gamma^2,
\\  & 6\Big( h^2\gamma^2+ \frac{3120}{829}\gamma^4 h^4 \Big) \le \Big(6+6\frac{3120}{829}\frac{1}{16}\Big)(h\gamma)^2\le 8 h^2\gamma^2,
\\ & 8 \gamma^{-2} L_K^2 h^2 +  4\gamma^{-2}h^4 \frac{6144}{829} L_K^2 \Big( (\gamma+L_K h)^2+ \frac{h^2 L_G^2}{4}\Big)+\gamma^{-2} h^2L_G^2 
\\ & \qquad \le 8 \gamma^{-2} L_K^2 h^2+ 4 (h\gamma)^2\frac{6144}{829}\frac{1}{16}\Big(\frac{25}{16}+\frac{1}{256}\Big)+(h\gamma)^2/4
\\ & \qquad \le \Big(16L_K h^2+3(h\gamma)^2\Big)\Big(\frac{L_K}{2\gamma^2}+1\Big), 
\\ &  4 h^4 \frac{6144}{829} L_K^2\le \frac{1536}{829} (h\gamma)^2\le 2 (h\gamma)^2. 
\end{align*}
Plugging these bounds back, we can bound
\begin{align*}
& \Big(\frac{L_K}{2\gamma^2}+1+\gamma^{-2}h^2L_G^2\Big)\hat{a}_{k+1}+b_{k+1}
\\ & \le \Big(\frac{L_K}{2\gamma^2}+1+\gamma^{-2}h^2L_G^2\Big)\Big((k+1)\sum_{i=1}^k 10 h^2\gamma^2\hat{a}_i +(k+1)\sum_{i=1}^k 8 h^2\gamma^2 b_i+ \mathbf{M}_1\Big)+ \mathbf{M}_2
\\ & +  (k+1) \sum_{i=0}^k  \Big(\Big(16L_K h^2+3(h\gamma)^2\Big)\Big(\frac{L_K}{2\gamma^2}+1\Big) \hat{a}_i + 2 (h\gamma)^2 b_i \Big)
+ L_G^2h^2\gamma^{-2} \hat{a}_{k+1}.
\end{align*}
Subtracting $L_G^2h^2\gamma^{-2} \hat{a}_{k+1}$ on both sides, using $a_{k+1}=(\frac{L_K}{2\gamma^2}+1)\hat{a}_{k+1}$ and $\gamma^{-2} h^2L_G^2\le 1/64$ which holds by \eqref{eq:condh}, we obtain
\begin{align*}
 a_{k+1}+b_{k+1}&
 \le \Big(\frac{L_K}{2\gamma^2}+\frac{65}{64}\Big)\Big((k+1)\sum_{i=1}^k 10 h^2\gamma^2\hat{a}_i +(k+1)\sum_{i=1}^k 8 h^2\gamma^2 b_i+ \mathbf{M}_1\Big)
\\ & \quad+  (k+1) \sum_{i=0}^k  \Big(16L_K h^2+3(h\gamma)^2\Big){a}_i
 +(k+1) \sum_{i=0}^k  2(h\gamma)^2 b_i+ \mathbf{M}_2
 \\ & \le (k+1)\sum_{i=1}^k \Big(10\cdot\frac{65}{64}h^2\gamma^2+16L_K h^2+3(h\gamma)^2\Big) {a}_i +\frac{65}{64}\Big)\mathbf{M}_1
 \\ & \quad+(k+1)\sum_{i=1}^k \Big(\Big(\frac{L_K}{2\gamma^2}+\frac{65}{64}\Big)8h^2\gamma^2+ 2 (h\gamma)^2 \Big) b_i+ \Big(\frac{L_K}{2\gamma^2} + \mathbf{M}_2
 \\ & \le (k+1)\sum_{i=1}^k \Big(\Big(16L_K h^2+\frac{421}{32}\gamma^2h^2\Big) {a}_i + \Big(\frac{L_K}{\gamma^2}4 h^2\gamma^2+ \frac{81}{8} \gamma^2h^2 \Big) b_i\Big)
 \\ & \quad + \Big(\frac{L_K}{2\gamma^2}+\frac{65}{64}\Big)\mathbf{M}_1 + \mathbf{M}_2
 \\ & \le  (k+1)\sum_{i=1}^k \hat{\lambda} ({a}_i+b_i) +\Big(\frac{L_K}{2\gamma^2}+\frac{65}{64}\Big)\mathbf{M}_1 + \mathbf{M}_2
\end{align*}
with
\begin{align*}
\hat{\lambda}=\Big(16 L_K +\frac{421}{32}\gamma^2\Big). 
\end{align*}
As in the $\mathcal{PG}$-splitting, we fix $l\in\mathbb{N}$ such that $l=\lfloor \hat{\lambda}^{-1/2} h^{-1}\rfloor$. Then, for $(k+1)\le l$, it holds
\begin{align*}
a_{k+1}+b_{k+1}\le \sum_{i=1}^k h \hat{\lambda}^{1/2} (a_i+b_i) + h^4 \mathbf{M}_3.
\end{align*}
where 
\begin{align*}
\mathbf{M}_3&=\hat{\lambda}^{-1}\Big(\frac{L_K}{2\gamma^2}+\frac{65}{64}\Big)L_G^2 \Big(\frac{2}{5}+  27\frac{\gamma^2 }{\kappa}\Big) d + \hat{\lambda}^{-1}35\gamma^{-2} L_K^2 L_G^2 \frac{d}{\kappa}+ \hat{\lambda}^{-1/2} 12\gamma^{-1}L_G^2 d 
\\ & + \hat{\lambda}^{-1}\gamma^{-2}\frac{18}{5} L_H(d^2+2d) + \hat{\lambda}^{-1} \gamma^{-2}\frac{18}{5}L_G^2(2\gamma^2d+2(L_K+L_G)d).
\end{align*}
Note that we used $\frac{21840}{829}\le 27$ and $\frac{28672}{829}\le 35$ to simplify the bounds of $\mathbf{M}_1$ and $\mathbf{M}_2$.
We observe that there exists a sequence $(c_k)_{k\in\mathbb{N}}$ satisfying $a_k+b_k\le c_k $ for $k\le h^{-1}\hat{\lambda}^{-1/2}$, $c_1=h^4\mathbf{M}_3$ and 
\begin{align*}
c_{k+1}=\sum_{i=1}^k h\sqrt{\hat{\lambda}}c_i+ h^4\mathbf{M}_3=c_k+ h \sqrt{\hat{\lambda}} c_k = (1+h\sqrt{\hat{\lambda}})^k c_1.
\end{align*}
Then, for $k+1\le l$
\begin{align*}
a_{k+1}+b_{k+1}\le c_{k+1}\le (1+h\sqrt{\hat{\lambda}})^k c_1\le e^{h k \sqrt{\hat{\lambda}}} c_1=e^{1}  h^4\mathbf{M}_3.
\end{align*}
By \eqref{eq:W2_stracc3} and since $lh> 2 \hat{\lambda}^{-1/2}$ by $l=\lfloor \hat{\lambda}^{-1/2} h^{-1}\rfloor$ 
\begin{align*}
\mathcal{W}_{2,\rho}(\mu,\tilde{\mu}_h)&\le \Big(1+\frac{1}{chl}\Big) \sqrt{2(a_l+b_l)} \le h^2  \Big(1+\frac{2\sqrt{\hat{\lambda}}}{c}\Big)\sqrt{2e^{1}\mathbf{M}_3}. 
\end{align*}
We note that for $c=\gamma/8$, it holds $ \hat{\lambda}^{-1/2}\le \frac{\sqrt{32}}{\sqrt{421}\gamma}= \frac{\sqrt{32}}{\sqrt{421}}\frac{1}{8} c^{-1}= \frac{1}{\sqrt{842}} c^{-1} $ and for $c=\kappa\gamma^{-1}/4$ it holds $\hat{\lambda}^{-1/2}\le \sqrt{1/16}\frac{1}{\sqrt{L_K}}\le \frac{1}{4}\frac{1}{\kappa \gamma^{-1}}\sqrt{\kappa \gamma^{-2}}\le  \frac{1}{4} \frac{1}{\sqrt{2}}\frac{1}{4}c^{-1}$.
Hence,  $\hat{\lambda}^{-1/2}\le (1/842)^{1/2}c^{-1}$. Using this bound and inserting $\mathbf{M}_3$, we bound 
\begin{align*}
\mathcal{W}_{2,\rho}(\mu,\tilde{\mu}_h)& \le h^2  \Big(\frac{1}{\sqrt{\hat{\lambda}}}+\frac{2}{c}\Big)\mathcal{C}\Big(L_G \sqrt{d}\sqrt{ \frac{L_K}{\gamma^2}+ \frac{\gamma^2}{\kappa}+ \frac{L_K}{\kappa}+\frac{L_K^2}{\kappa \gamma^2}+ \frac{\sqrt{L_K}}{\gamma}}+ \sqrt{L_H}\frac{d}{\gamma}\Big)
\\ & \le h^2  \frac{1}{c}\mathcal{C}\Big(L_G \sqrt{d}\sqrt{ \frac{L_K}{\gamma^2}+ \frac{\gamma^2}{\kappa}+ \frac{L_K}{\kappa}+\frac{L_K^2}{\kappa \gamma^2}}+ \sqrt{L_H}\frac{d}{\gamma}\Big),
\end{align*}
where $\mathcal{C}\in\mathbb{R}_+$ is some number which varies from line to line. This bound in combination with \eqref{eq:equivdist} yields the result.
\end{proof}

\begin{proof}[Proof of Theorem~\ref{thm:complex_PGP}]
The result follows immediately by Theorem~\ref{thm:conv_PGP2}, Theorem~\ref{thm:stracc_PGP} and
\begin{align*}
\mathcal{W}_{2,\rho}(\mu, \nu \tilde{\pi}_h^k)&\le \mathcal{W}_{2,\rho}(\mu, \tilde{\mu}_h)+\mathcal{W}_{2,\rho}(\tilde{\mu}_h, \nu \tilde{\pi}_h^k).
\end{align*}
\end{proof}

\section{Numerical experiments and discussion}

We implemented the $\mathcal{PGP}$ sampler for two models and compared its long-time behaviour to the one of the $\mathcal{OBABO}$ sampler \cite{Mo2021}. In the first model, we considered the potential $U(x)=\frac{1}{2}x^T Kx+G(x)$
with
\begin{align*}
K=\begin{pmatrix}
1 & 0 \\ 0 & 10
\end{pmatrix}, \qquad \text{and} \qquad G(x)=\frac{1}{4}\Big(\frac{1}{2}|x_1|^2+\frac{1}{2}|x_2|^2 +\frac{1}{2}\sin(x_1+x_2)\Big). 
\end{align*}
Note that $G$ is convex, $L_G=1/2$ and $L_K=10$.
In the second model, we considered a logistic type potential with the same matrix $K$ as in the first model and
\begin{align*}
G(x)=\frac{1}{10}\sum_{i=1}^2 \log(1+e^{a_i^T x}) ,
\end{align*}
where $a_1=(1,0)$ and $a_2=(0,2)$. In this model $G$ is convex with $L_G\le 4/10$.

Then for both models, the parameter choice $\gamma=2$ and $h=0.01$ satisfies \eqref{eq:condh2}. In Figure~\ref{fig:1} and Figure~\ref{fig:2}, contraction of the $\mathcal{PGP}$ splitting is illustrated and compared to the behaviour of the $\mathcal{OBABO}$ splitting scheme. Averaging the distance over 200 runs we observe exponential contraction for both schemes with the same rate which indicates that the novel scheme behaves comparably well in experiments. 

Additionally, the numerical experiments are consistent with the theoretical analysis on the long-time behaviour. Together with the theoretical results on the complexity guarantees, these results offer a solid foundation for this novel kinetic Langevin sampler using the exact harmonic Langevin integrator.

Due to its specific structure, the sampler appears particularly well suited for Bayesian sampling problems with Gaussian prior. Extending the analysis to stochastic-gradient variants would be an interesting direction for future work.

\begin{figure}[t!]
	\centering
        \includegraphics[width=0.45\textwidth]{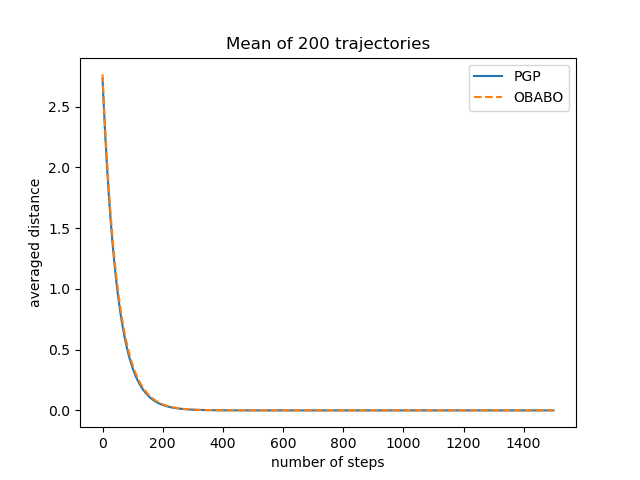}
	\includegraphics[width=0.45\textwidth]{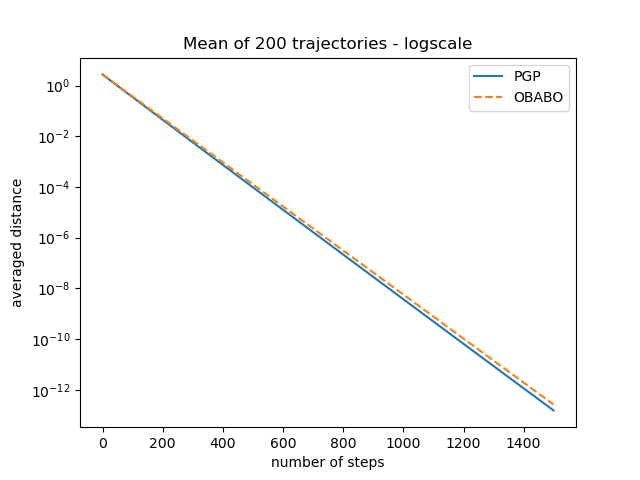}
        \caption{Potential with oscillations: Blue lines on the plot are the contraction of the averaged distance (over 200 samples) of two synchronously coupled trajectories of the PGP splitting schemes, whereas the orange lines illustrate the contraction of the OBABO splitting scheme. For both schemes the two coupled trajectories are initialized at $((1,1),(1,1))$ and $((-1,-1),(-1,-1))$, respectively.}
        \label{fig:1}
\end{figure}

\begin{figure}
	\centering
        \includegraphics[width=0.45\textwidth]{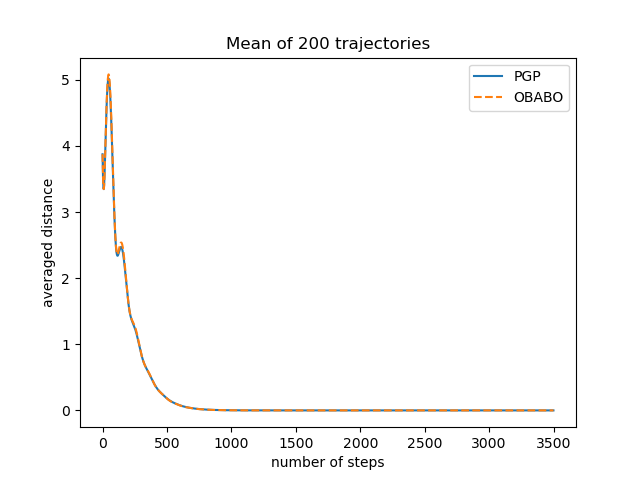}
	\includegraphics[width=0.45\textwidth]{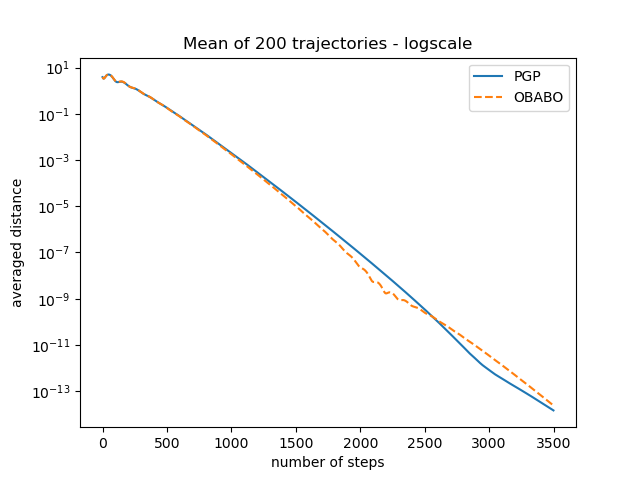}
        \caption{Logistic type potential: Blue lines on the plot are the contraction of the averaged distance (over 200 samples) of two synchronously coupled trajectories of the PGP splitting schemes, whereas the orange lines illustrate the contraction of the OBABO splitting scheme. For both schemes the two coupled trajectories are initialized at $((1,1),(1,1))$ and $((-1,-1),(-1,-1))$, respectively.}
        \label{fig:2}
\end{figure}

\appendix

\section{Proofs on the numerical scheme} \label{sec:proof_num}


\begin{proof}[Proof of Proposition~\ref{prop:numscheme}]
By Assumption~\ref{ass:diag}, the components do not interact in $\mathcal{G}$ and it is sufficient to consider the step $\mathcal{G}$ component wise.
Therefore, without loss of generality, it is sufficient to consider the case $d=1$ and $K=\kappa>0$. 
Then, by \eqref{eq:exactsol_d} $(\tilde{X}_h,\tilde{V}_h)$ at time $h>0$ is given by
\begin{align} \label{eq:exactsol}
\begin{pmatrix}
\tilde{X}_h \\ \tilde{V}_h
\end{pmatrix}
= e^{Ah} 
\begin{pmatrix}
\tilde{X}_0 \\ \tilde{V}_0
\end{pmatrix}+ \sqrt{2\gamma} \int_0^h e^{A(h-s)}  \rmd \begin{pmatrix}
 0
\\  B_s
\end{pmatrix} \qquad \text{with }A=\begin{pmatrix}
0 & 1 \\ -\kappa & -\gamma
\end{pmatrix}.
\end{align} 
To compute the matrix exponential $e^{Ah}$, we have to distinguish three cases depending on the different types of eigenvalues of the matrix $A$. For $\frac{\gamma^2}{4}-\kappa>0$, we obtain the overdamped case with two different real eigenvalues. For $\frac{\gamma^2}{4}-\kappa<0$, we obtain the underdamped case with two different complex-valued eigenvalues. In the third case, i.e., $\frac{\gamma^2}{4}-\kappa=0$, $A$ has one real-valued eigenvalue.

In the case of two different eigenvalues, the eigenvalues are given by $\lambda_{1}= -\gamma/2+\sqrt{\frac{\gamma^2}{4}-\kappa}$ and $\lambda_{2}= -\gamma/2-\sqrt{\frac{\gamma^2}{4}-\kappa}$ and the matrix $A$ can be written as
\begin{align*}
A=P \Lambda P^{-1} \quad \text{with } \Lambda=\begin{pmatrix}
\lambda_1  & 0 \\ 0 & \lambda_2
\end{pmatrix}, \  P=\begin{pmatrix}
1 & 1 \\ \lambda_1 & \lambda_2
\end{pmatrix},  \ P^{-1}=\frac{1}{\lambda_1-\lambda_2}\begin{pmatrix}
-\lambda_2 & 1 \\ \lambda_1 & -1
\end{pmatrix}.  
\end{align*}
Then, for $t\ge 0$
\begin{align*}
e^{At}&=P \begin{pmatrix}
e^{\lambda_1 t} & 0 \\ 0 & e^{\lambda_2 t} 
\end{pmatrix}P^{-1}
\\ & = \frac{e^{-\frac{\gamma}{2} t}}{2 \sqrt{\frac{\gamma^2}{4} - \kappa}} \begin{pmatrix}
-\lambda_2 e^{\sqrt{\frac{\gamma^2}{4} - \kappa}t} +e^{-\sqrt{\frac{\gamma^2}{4} - \kappa}t} & e^{\sqrt{\frac{\gamma^2}{4} - \kappa}t}-e^{-\sqrt{\frac{\gamma^2}{4}- \kappa}t}
\\ \lambda_2 \lambda_1(e^{\sqrt{\frac{\gamma^2}{4} - \kappa}t}-e^{-\sqrt{\frac{\gamma^2}{4} - \kappa}t}) & \lambda_1 e^{\sqrt{\frac{\gamma^2}{4} - \kappa}t}-\lambda_2 e^{-\sqrt{\frac{\gamma^2}{4} - \kappa}t}
\end{pmatrix} .
\end{align*}
Inserting the different eigenvalues and using hyperbolic and trigonometric functions, respectively, we obtain
in the overdamped case  ($\frac{\gamma^2}{4}-\kappa>0$)
\begin{align*}
e^{At}&=e^{-\frac{\gamma}{2}t} \begin{pmatrix}
\frac{1}{\omega}\sinh(\frac{\gamma}{2}\omega t)+\cosh(\frac{\gamma}{2}\omega t) & \frac{2}{\gamma\omega}\sinh(\frac{\gamma}{2}\omega t)
\\ -\kappa\frac{2}{\gamma\omega} \sinh(\frac{\gamma}{2}\omega t) & -\frac{1}{\omega}\sinh(\frac{\gamma}{2}\omega t)+\cosh(\frac{\gamma}{2}\omega t)
\end{pmatrix}
\end{align*}
and in the underdamped case ($\frac{\gamma^2}{4}-\kappa<0$)
\begin{align*}
e^{At}&=e^{-\frac{\gamma}{2}t} \begin{pmatrix}
\frac{1}{\omega}\sin(\frac{\gamma}{2}\omega t)+\cos(\frac{\gamma}{2}\omega t) & \frac{2}{\gamma\omega}\sin(\frac{\gamma}{2}\omega t)
\\ -\kappa\frac{2}{\gamma\omega} \sin(\frac{\gamma}{2}\omega t) & -\frac{1}{\omega}\sin(\frac{\gamma}{2}\omega t)+\cos(\frac{\gamma}{2}\omega t)
\end{pmatrix},
\end{align*}
with $\omega= \sqrt{|1-4\kappa/\gamma^2|}$.

In the case of one real eigenvalue for $A$, the eigenvalue is given by $\lambda=-\gamma/2$ and the matrix $A$ satisfies
\begin{align*}
A=P \Lambda P^{-1} \quad \text{with } \Lambda=\begin{pmatrix}
\lambda  & -\lambda \\ 0 & \lambda
\end{pmatrix}, \quad  P=\begin{pmatrix}
1 & 1 \\ \lambda & 0
\end{pmatrix} \quad \text{and } P^{-1}=\frac{1}{\lambda}\begin{pmatrix}
0 & 1 \\ \lambda & -1
\end{pmatrix}.  
\end{align*}
Then, for $t\ge 0$
\begin{align*}
e^{At}&=P \begin{pmatrix}
e^{\lambda t} & 0 \\ 0 & e^{\lambda t} 
\end{pmatrix}P^{-1}
 = e^{-\frac{\gamma}{2} t} \begin{pmatrix}
\frac{\gamma}{2}t+1& t
\\ -\kappa t & -\frac{\gamma}{2}t +1
\end{pmatrix} .
\end{align*} 
Setting $t=h$, we obtain $\mathbf{A}(h)$.
To compute $\mathbf{B}(h)$, we distinguish again the overdamped, underdamped and critical case. For the overdamped case we can rewrite the noise term in \eqref{eq:exactsol} by
\begin{align*}
 \int_0^h e^{A(h-s)}  &\rmd \begin{pmatrix}
0 \\ B_s
\end{pmatrix}  = \begin{pmatrix}
\int_0^h e^{-\gamma/2 (h-s)} \frac{2}{\gamma \omega}\sinh(\frac{\gamma}{2}\omega (h-s)) \rmd B_s 
\\ \int_0^h e^{-\gamma/2 (h-s)}  (\frac{-1}{\omega}\sinh(\frac{\gamma}{2}\omega (h-s))+\cosh(\frac{\gamma}{2}\omega (h-s))) \rmd B_s
\end{pmatrix}
\\  & = \begin{pmatrix} \int_0^h  \frac{1}{\gamma \omega}\Big( e^{-\frac{\gamma}{2}(1-\omega)(h-s)}-e^{-\frac{\gamma}{2}(1+\omega)(h-s)}\Big)\rmd B_s
\\ \int_0^h  \frac{1 }{2}\Big(e^{-\frac{\gamma}{2}(1+\omega)(h-s)}\frac{1+\omega}{\omega}-e^{-\frac{\gamma}{2}(1-\omega)(h-s)}\frac{1-\omega}{\omega}\Big) \rmd B_s 
\end{pmatrix}=:\begin{pmatrix}
\mathcal{Z}_1 \\ \mathcal{Z}_2
\end{pmatrix}.
\end{align*}
To represent the noise term via standard normally distributed random variables, we consider for some $a,b\in \mathbb{R}$
\begin{align} \label{eq:ZaZb}
Z_a:=\int_0^h e^{-a(h-s)} \rmd B_s, \qquad \text{and } \qquad  Z_b:=\int_0^h e^{-b(h-s)} \rmd B_s
\end{align}
and observe due to Ito's formula
\begin{align*}
&\sigma_a^2:= \mathbb{E}[|Z_a|^2]= \frac{1-e^{-2ha}}{2a}, \quad  \sigma_b^2:= \mathbb{E}[|Z_b|^2]=\frac{1-e^{-2hb}}{2b}, 
\\ & c:=\mathbb{E}[Z_a Z_b]= \frac{1-e^{-(a+b)h}}{a+b}.
\end{align*}
Let $\xi,\zeta\sim \mathcal{N}(0,1)$ be two independent standard normally distributed random variables.
Then, $\tilde{Z}_a$ and $\tilde{Z}_b$ given by
\begin{align} \label{eq:tildeZaZb}
\tilde{Z}_a=\sigma_a \xi, \qquad \text{and} \qquad  \tilde{Z}_b=\sigma_b\Big(\frac{c}{\sigma_a \sigma_b}\xi +\sqrt{1-\Big(\frac{c}{\sigma_a \sigma_b}\Big)^2}\zeta\Big)
\end{align}
satisfy $\mathbb{E}[|\tilde{Z}_a|^2]=\sigma_a^2$, $\mathbb{E}[|\tilde{Z}_b|^2]=\sigma_b^2$ and $\mathbb{E}[|\tilde{Z}_a\tilde{Z}_b]=c$ and are equivalent to $Z_a$ and $Z_b$.
Setting $a=\frac{\gamma}{2}(1+\omega)$ and $b=\frac{\gamma}{2}(1-\omega)$, we note $a+b=\gamma$ and we obtain that $(\mathcal{Z}_1, \mathcal{Z}_2)$ is equivalent to 
\begin{align*}
\begin{pmatrix}
\frac{1}{\gamma \omega} \Big(
\frac{(1+\omega)(1-e^{-\gamma h} )-(1-e^{-\gamma(1+\omega)h})}{\sqrt{\gamma(1-e^{-\gamma(1+\omega) h})(1+\omega) } }\xi
+ \sqrt{\frac{1-e^{-\gamma(1-\omega)h}}{\gamma(1-\omega)}-\frac{(1-e^{-\gamma h})^2(1+\omega)}{\gamma(1-e^{-\gamma(1+\omega)h})}} \zeta \Big)
\\
\frac{1+\omega}{2\omega}\Big(\frac{(1-e^{-\gamma(1+\omega)h})-(1-\omega)(1-e^{-\gamma h})}{\sqrt{\gamma(1-e^{-\gamma(1+\omega)h})(1+\omega)}}\Big)\xi
- \frac{1-\omega}{2\omega} \sqrt{\frac{1-e^{-\gamma(1-\omega)h}}{\gamma(1-\omega)}-\frac{(1-e^{-\gamma h})^2(1+\omega)}{\gamma(1-e^{-\gamma(1+\omega)h})}} \zeta 
\end{pmatrix}
\end{align*}
which multiplied by $\sqrt{2\gamma}$ provides $\mathbf{B}(h)$ in the overdamped case.

Analogously for the underdamped case, we observe that the noise term in \eqref{eq:exactsol} can be written as
\begin{align*}
\begin{pmatrix}
\mathcal{Z}_1 \\ \mathcal{Z}_2
\end{pmatrix}
= \begin{pmatrix}
\int_0^h e^{-\gamma/2 (h-s)}  \frac{2}{\gamma \omega}\sin(\frac{\gamma}{2}\omega (h-s)) \rmd B_s
\\ \int_0^h e^{-\gamma/2 (h-s)}  (\frac{-1}{\omega}\sin(\frac{\gamma}{2}\omega (h-s))+\cos(\frac{\gamma}{2}\omega (h-s))) \rmd B_s 
\end{pmatrix}
\end{align*}
where $(B_t)_{t\ge 0}$ is a standard Brownian motion. 
To represent the noise term via standard normally distributed random variables, we consider for some $a\in \mathbb{R}$
\begin{align*} 
&Z_{sin}:=\int_0^h e^{-\frac{\gamma}{2}(h-s)}\sin\Big(\frac{\gamma\omega}{2}(h-s)\Big) \rmd B_s,
\\ & Z_{cos}:=\int_0^h e^{-\frac{\gamma}{2}(h-s)}\cos\Big(\frac{\gamma\omega}{2}(h-s)\Big)  \rmd B_s
\end{align*}
and observe due to Ito's formula
\begin{align*}
\sigma_{sin}^2 & := \mathbb{E}[|Z_{sin}|^2]
 =  \int_0^h e^{-\gamma(h-s)} \frac{(e^{i\frac{\gamma\omega}{2}(h-s)}-e^{-i\frac{\gamma\omega}{2}(h-s)})^2}{(2i)^2} \rmd s 
\\ & = \frac{-1}{4}\Big(\frac{1-e^{-\gamma(1-i\omega) h}}{\gamma(1-i\omega)}-2\frac{1-e^{-\gamma h}}{\gamma}+\frac{1-e^{-\gamma(1+i\omega) h}}{\gamma(1+i\omega)}\Big) 
\\ & =\frac{1}{2(1+\omega^2)\gamma}\Big(-1+e^{-\gamma h} \cos(\omega \gamma h)- \omega e^{-\gamma h}\sin(\omega\gamma h)+(1-e^{-\gamma  h} )(1+\omega^2)\Big) , 
\\   \sigma_{cos}^2 & := \mathbb{E}[|Z_{cos}|^2]
 = \int_0^h e^{-\gamma(h-s)} \Big(1-\sin\Big(\frac{\gamma\omega}{2}(h-s)\Big)^2\Big) \rmd s
\\ & =\frac{1}{2(1+\omega^2)\gamma}\Big(1-e^{-\gamma h} \cos(\omega \gamma h)+ \omega e^{-\gamma h}\sin(\omega\gamma h)+(1-e^{-\gamma  h} )(1+\omega^2)\Big) ,   
\\  c & :=\mathbb{E}[Z_{sin} Z_{cos}] 
= \int_0^h \frac{1}{4 i}(e^{-\gamma(1-i\omega)(h-s)}-e^{-\gamma(1+i\omega)(h-s)})\rmd s
\\ & = \frac{1}{2(1+\omega^2)\gamma}\Big(\omega -e^{-\gamma h} (\sin(\gamma \omega h)+\cos(\gamma \omega h)\omega)\Big) .
\end{align*}
As for the overdamped case, let $\xi,\zeta\sim \mathcal{N}(0,1)$ be two independent standard normally distributed random variables.
Then, $\tilde{Z}_{sin}$ and $\tilde{Z}_{cos}$ given by
\begin{align*} 
\tilde{Z}_{sin}=\sigma_{sin} \xi, \qquad \text{and } \tilde{Z}_{cos}=\sigma_{cos}\Big(\frac{c}{\sigma_{sin} \sigma_{cos}}\xi +\sqrt{1-\Big(\frac{c}{\sigma_{sin} \sigma_{cos}}\Big)^2}\zeta\Big)
\end{align*}
satisfy $\mathbb{E}[|\tilde{Z}_{sin}|^2]=\sigma_{sin}^2$, $\mathbb{E}[|\tilde{Z}_{cos}|^2]=\sigma_{cos}^2$ and $\mathbb{E}[|\tilde{Z}_{sin}\tilde{Z}_{cos}]=c$. Hence, they are equivalent to $Z_{sin}$ and $Z_{cos}$. Then, the noise term $(\mathcal{Z}_1,\mathcal{Z}_2)$ is equivalent to 
\begin{align*}
\begin{pmatrix}
\frac{2}{\gamma \omega}\sigma_{sin} \xi
\\ \Big(-\frac{\sigma_{sin}}{\omega}+ \frac{c}{\sigma_{sin}}\Big) \xi+ \sqrt{\sigma_{cos}^2-\frac{c^2}{\sigma_{sin}^2}}\zeta
\end{pmatrix}
\end{align*}
which multiplied by $\sqrt{2\gamma}$ gives $\mathbf{B}(h)$ for the underdamped case.


In the critical case, the noise term of \eqref{eq:exactsol}, is given by 
\begin{align*}
\begin{pmatrix}
\mathcal{Z}_1 \\ \mathcal{Z}_2
\end{pmatrix}
= 
\begin{pmatrix}
\int_0^h \sqrt{2\gamma}e^{-\frac{\gamma}{2}(h-s)}(h-s) \rmd B_s \\\int_0^h \sqrt{2\gamma}e^{-\frac{\gamma}{2}(h-s)}(-\frac{\gamma}{2}(h-s)+1) \rmd B_s
\end{pmatrix}.
\end{align*}
To express the noise term via standard normally distributed random variables, we consider
\begin{align*}
Z_1:=\int_0^h e^{-\frac{\gamma}{2}(h-s)}(h-s)\rmd B_s, \qquad \text{and }\qquad Z_2:=\int_0^h e^{-\frac{\gamma}{2}(h-s)}\rmd B_s.
\end{align*}
Then, 
\begin{align*}
\sigma_1^2&:= \mathbb{E}[|Z_1|^2]=\int_0^h e^{-\gamma(h-s)} (h-s)^2 \rmd s = \frac{-e^{-\gamma h} h^2 \gamma^2 - 2e^{-\gamma h}h\gamma+2(1-e^{-\gamma h})}{\gamma^3},
\\ \sigma_2^2 &:= \mathbb{E}[|Z_2|^2]=\int_0^h e^{-\gamma(h-s)} \rmd s = \frac{1-e^{-\gamma h}}{\gamma},
\\ c&:= \mathbb{E}[Z_1Z_2]=\int_0^h e^{-\gamma(h-s)} (h-s) \rmd s = \frac{1-e^{-\gamma h}-h\gamma e^{-\gamma h}}{\gamma^2}.
\end{align*}
As in the two previous cases, let $\xi, \zeta\sim \mathcal{N}(0,1)$ be two independent standard normally distributed random variables. Then $\tilde{Z}_1$ and $\tilde{Z}_2$ given by
\begin{align*}
\tilde{Z}_1= \sigma_1 \xi, \quad \text{and } \tilde{Z}_2= \sigma_2\Big(\frac{c}{\sigma_1 \sigma_2}\xi+\sqrt{1-\Big(\frac{c}{\sigma_1 \sigma_2}\Big)^2}\zeta\Big)
\end{align*}
satisfy $\mathbb{E}[|\tilde{Z}_1|^2]=\sigma_1^2$, $\mathbb{E}[|\tilde{Z}_2|^2]=\sigma_2^2$ and $\mathbb{E}[|\tilde{Z}_1\tilde{Z}_2]=c$. Hence, they are equivalent to $Z_1$ and $Z_2$.
Therefore, the noise term $(\mathcal{Z}_1,\mathcal{Z}_2)$ is equivalent to 
\begin{align*}
\begin{pmatrix}
\sigma_{1} \xi
\\ \Big(-\frac{\gamma\sigma_{1}}{2}+ \frac{c}{\sigma_{1}}\Big) \xi+\sqrt{\sigma_{2}^2-\frac{c^2}{\sigma_{1}^2}}\zeta
\end{pmatrix}
\end{align*}
which gives $\mathbf{B}(h)$ for the critical case and concludes the proof.

\end{proof}

\section{Trapezoidal rule}

\begin{lemma}[Trapezoidal rule]\label{lem:trapez}
Let $(X_t,V_t)_{t\ge 0}$ solving \eqref{eq:LD}. Suppose Assumption~\ref{ass:U} and Assumption~\ref{ass_GHess} hold true. Then, for $h>0$ and $t\ge 0$ it holds
\begin{align*}
&\frac{h}{2}[\nabla G(X_t)+\nabla G(X_{t+h})]-\int_0^h \nabla G(X_{t+s})\rmd s
\\ & = \frac{h}{2}\int_0^h \int_0^r \nabla^2 G(X_{t+u})\sqrt{2\gamma}\rmd B_{t+u} \rmd r-\int_0^h \int_0^s \int_0^r \nabla^2 G(X_{t+u})\sqrt{2\gamma}\rmd B_{t+u} \rmd r \rmd s
\\ & +\int_0^h \Big(  \frac{h}{2} \int_0^r \nabla^3 G(X_{t+u})[V_{t+u},V_{t+u}] \rmd u - \int_0^s \int_0^r \nabla^3 G(X_{t+u})[V_{t+u},V_{t+u}] \rmd u \rmd r\Big) \rmd s 
\\ & + \frac{h}{2}\int_0^h \int_0^r \nabla^2 G(X_{t+u})(-\gamma V_{t+u}-KX_{t+u}-\nabla G(X_{t+u})) \rmd u \rmd r
\\ & \qquad -\int_0^h \int_0^s \int_0^r  \nabla^2 G(X_{t+u})(-\gamma V_{t+u}-KX_{t+u}-\nabla G(X_{t+u})) \rmd u \rmd r \rmd s.
\end{align*}
\end{lemma}

\begin{proof}[Proof of Lemma~\ref{lem:trapez}]
By Ito-Taylor expansion (see e.g. \cite[Theorem 5.5.1]{KlPl1992}), it holds for $t,s\ge 0$
\begin{align*}
\int_0^h\nabla  G(X_{t+s})\rmd s &= \int_0^h\Big( \nabla G(X_t)+\int_0^s \mathcal{L}^0 \nabla G(X_{t+r})\rmd r + \int_0^s \mathcal{L}^1 \nabla G(X_{t+r}) \rmd B_{t+r}\Big)\rmd s
\\ & = h\nabla G(X_t)+ \int_0^h\int_0^s \Big(\mathcal{L}^0 \nabla G(X_{t}) +\int_0^r\mathcal{L}^0 \mathcal{L}^0\nabla G(X_{t+u})\rmd u 
\\ & \qquad  +\int_0^r\mathcal{L}^1 \mathcal{L}^0\nabla G(X_{t+u})\rmd B_{t+u}\Big)\rmd r \rmd s
\\ & = h\nabla G(X_t)+\frac{h^2}{2}\mathcal{L}^0 \nabla G(X_{t}) + \int_0^h\int_0^s \int_0^r\mathcal{L}^0 \mathcal{L}^0\nabla G(X_{t+u})\rmd u \rmd r \rmd s
\\ & \qquad  +\int_0^h\int_0^s \int_0^r\mathcal{L}^1 \mathcal{L}^0\nabla G(X_{t+u})\rmd B_{t+u}\rmd r \rmd s,
\end{align*}
where the operator $\mathcal{L}^0$ and $\mathcal{L}^1$ are given by 
\begin{align*}
\mathcal{L}^0= v \nabla_x + (-\gamma v- K x-\nabla G(x))\nabla_v +\gamma \nabla_v^2, \qquad \mathcal{L}^1=\sqrt{2\gamma}\nabla_v.
\end{align*}
Note that in the second step, we used $\mathcal{L}^1 \nabla G(X_{t+r})=0$.
Similarly, 
\begin{align*}
\nabla G(X_{t+h})&= \nabla G(X_t)+ \int_0^h \mathcal{L}^0 \nabla G(X_{t+r})\rmd r + \int_0^s \mathcal{L}^1 \nabla G(X_{t+r}) \rmd B_{t+r}
\\ & = \nabla G(X_t)+ \int_0^h \Big(\mathcal{L}^0 \nabla G(X_{t}) +\int_0^r\mathcal{L}^0 \mathcal{L}^0\nabla G(X_{t+u})\rmd u 
\\ & \quad +\int_0^r\mathcal{L}^1 \mathcal{L}^0\nabla G(X_{t+u})\rmd B_{t+u}\Big)\rmd r. 
\end{align*}
Hence, 
\begin{align*}
T&=\frac{h}{2}[\nabla G(X_t)+\nabla G(X_{t+h})]
\\ & =h\nabla G(X_t)+ \frac{h^2}{2}\mathcal{L}^0 \nabla G(X_{t}) +\frac{h}{2}\int_0^h \int_0^r\mathcal{L}^0 \mathcal{L}^0\nabla G(X_{t+u})\rmd u\rmd r 
\\ & \qquad +\frac{h}{2}\int_0^h\int_0^r\mathcal{L}^1 \mathcal{L}^0\nabla G(X_{t+u})\rmd B_{t+u}\rmd r. 
\end{align*}
Hence, 
\begin{align*}
&T-\int_0^h\nabla  G(X_{t+s})\rmd s
\\ & = \frac{h}{2}\int_0^h \int_0^r\mathcal{L}^0 \mathcal{L}^0\nabla G(X_{t+u})\rmd u\rmd r  +\frac{h}{2}\int_0^h\int_0^r\mathcal{L}^1 \mathcal{L}^0\nabla G(X_{t+u})\rmd B_{t+u}\rmd r
\\ &- \int_0^h\int_0^s \int_0^r\mathcal{L}^0 \mathcal{L}^0\nabla G(X_{t+u})\rmd u \rmd r \rmd s  -\int_0^h\int_0^s \int_0^r\mathcal{L}^1 \mathcal{L}^0\nabla G(X_{t+u})\rmd B_{t+u}\rmd r \rmd s.
\end{align*}
Using that
\begin{align*}
&\mathcal{L}^0 \nabla G(X_s)=\nabla^2 G(X_s)V_s, \qquad \mathcal{L}^1 (\nabla^2 G(X_s)V_s)= \nabla^2G(X_s)\sqrt{2\gamma}, \qquad \text{and}
\\ & \mathcal{L}^0 (\nabla^2 G(X_s)V_s)= \nabla^3 G(X_s) [V_s,V_s]+ \nabla^2 G(X_s) (-\gamma V_s-K X_s-\nabla G(X_s))
\end{align*}
gives the desired result.
\end{proof}

\section*{Acknowledgments}
The author wants to thank Andreas Eberle, Francis Lörler and Stefan Oberdörster for fruitful discussions during her research visit in Bonn.

\bibliographystyle{plain}
\bibliography{biblio}

\begin{thebibliography}{10}

\bibitem{AbBuHi2017}
{\sc M.~Ableidinger, E.~Buckwar, and H.~Hinterleitner}, {\em A stochastic
  version of the jansen and rit neural mass model: analysis and numerics}, The
  Journal of Mathematical Neuroscience, 7 (2017), p.~8.

\bibitem{AnDeDoJo2003}
{\sc C.~Andrieu, N.~De~Freitas, A.~Doucet, and M.~I. Jordan}, {\em An
  introduction to mcmc for machine learning}, Machine learning, 50 (2003),
  pp.~5--43.

\bibitem{BaCaGu2008}
{\sc D.~Bakry, P.~Cattiaux, and A.~Guillin}, {\em Rate of convergence for
  ergodic continuous markov processes: Lyapunov versus poincar{\'e}}, Journal
  of Functional Analysis, 254 (2008), pp.~727--759.

\bibitem{BaGeLe2014}
{\sc D.~Bakry, I.~Gentil, and M.~Ledoux}, {\em Analysis and geometry of
  {M}arkov diffusion operators}, vol.~348 of Grundlehren der mathematischen
  Wissenschaften [Fundamental Principles of Mathematical Sciences], Springer,
  Cham, 2014, \url{https://doi.org/10.1007/978-3-319-00227-9},
  \url{https://doi.org/10.1007/978-3-319-00227-9}.

\bibitem{BoEbZi2020}
{\sc N.~Bou-Rabee, A.~Eberle, and R.~Zimmer}, {\em Coupling and convergence for
  {H}amiltonian {M}onte {C}arlo}, Ann. Appl. Probab., 30 (2020),
  pp.~1209--1250, \url{https://doi.org/10.1214/19-AAP1528},
  \url{https://doi.org/10.1214/19-AAP1528}.

\bibitem{BoSc2023}
{\sc N.~Bou-Rabee and K.~Schuh}, {\em Convergence of unadjusted hamiltonian
  monte carlo for mean-field models}, Electronic Journal of Probability, 28
  (2023), pp.~1--40.

\bibitem{BrLi1976}
{\sc H.~J. Brascamp and E.~H. Lieb}, {\em On extensions of the
  {B}runn-{M}inkowski and {P}r\'ekopa-{L}eindler theorems, including
  inequalities for log concave functions, and with an application to the
  diffusion equation}, J. Functional Analysis, 22 (1976), pp.~366--389,
  \url{https://doi.org/10.1016/0022-1236(76)90004-5},
  \url{https://doi.org/10.1016/0022-1236(76)90004-5}.

\bibitem{BuPa2007}
{\sc G.~Bussi and M.~Parrinello}, {\em Accurate sampling using langevin
  dynamics}, Physical Review E—Statistical, Nonlinear, and Soft Matter
  Physics, 75 (2007), p.~056707.

\bibitem{CaDuMoSt2023}
{\sc E.~Camrud, A.~Durmus, P.~Monmarch{\'e}, and G.~Stoltz}, {\em Second order
  quantitative bounds for unadjusted generalized hamiltonian monte carlo},
  arXiv preprint arXiv:2306.09513,  (2023).

\bibitem{CaLuWa2023}
{\sc Y.~Cao, J.~Lu, and L.~Wang}, {\em On explicit l 2-convergence rate
  estimate for underdamped langevin dynamics}, Archive for Rational Mechanics
  and Analysis, 247 (2023), p.~90.

\bibitem{CaGu2009}
{\sc P.~Cattiaux and A.~Guillin}, {\em Trends to equilibrium in total variation
  distance}, Ann. Inst. Henri Poincar\'e{} Probab. Stat., 45 (2009),
  pp.~117--145, \url{https://doi.org/10.1214/07-AIHP152},
  \url{https://doi.org/10.1214/07-AIHP152}.

\bibitem{ChMo2025}
{\sc M.~Chak and P.~Monmarch{\'e}}, {\em Reflection coupling for unadjusted
  generalized hamiltonian monte carlo in the nonconvex stochastic gradient
  case}, IMA Journal of Numerical Analysis,  (2025), p.~draf045.

\bibitem{ChGaJi2023}
{\sc Y.~Chen, K.~Gatmiry, and M.~Jiang}, {\em When does metropolized
  hamiltonian monte carlo provably outperform metropolis-adjusted langevin
  algorithm?}, arXiv preprint arXiv:2304.04724,  (2023).

\bibitem{ChVe2022}
{\sc Z.~Chen and S.~S. Vempala}, {\em Optimal convergence rate of hamiltonian
  monte carlo for strongly logconcave distributions}, Theory of Computing, 18
  (2022), pp.~1--18.

\bibitem{ChChBaJo2018}
{\sc X.~Cheng, N.~S. Chatterji, P.~L. Bartlett, and M.~I. Jordan}, {\em
  Underdamped langevin mcmc: A non-asymptotic analysis}, in Conference on
  learning theory, PMLR, 2018, pp.~300--323.

\bibitem{Da2017}
{\sc A.~S. Dalalyan}, {\em Theoretical guarantees for approximate sampling from
  smooth and log-concave densities}, J. R. Stat. Soc. Ser. B. Stat. Methodol.,
  79 (2017), pp.~651--676, \url{https://doi.org/10.1111/rssb.12183},
  \url{https://doi.org/10.1111/rssb.12183}.

\bibitem{DaRi2020}
{\sc A.~S. Dalalyan and L.~Riou-Durand}, {\em On sampling from a log-concave
  density using kinetic {L}angevin diffusions}, Bernoulli, 26 (2020),
  pp.~1956--1988, \url{https://doi.org/10.3150/19-BEJ1178}.

\bibitem{DePaBoDo2021}
{\sc G.~Deligiannidis, D.~Paulin, A.~Bouchard-C\^{o}t\'{e}, and A.~Doucet},
  {\em Randomized {H}amiltonian {M}onte {C}arlo as scaling limit of the bouncy
  particle sampler and dimension-free convergence rates}, Ann. Appl. Probab.,
  31 (2021), pp.~2612--2662, \url{https://doi.org/10.1214/20-aap1659},
  \url{https://doi.org/10.1214/20-aap1659}.

\bibitem{DoMoSc2009}
{\sc J.~Dolbeault, C.~Mouhot, and C.~Schmeiser}, {\em Hypocoercivity for
  kinetic equations with linear relaxation terms}, Comptes Rendus.
  Math{\'e}matique, 347 (2009), pp.~511--516.

\bibitem{DoMoSc2015}
{\sc J.~Dolbeault, C.~Mouhot, and C.~Schmeiser}, {\em Hypocoercivity for linear
  kinetic equations conserving mass}, Transactions of the American Mathematical
  Society, 367 (2015), pp.~3807--3828.

\bibitem{DuMo2017}
{\sc A.~Durmus and E.~Moulines}, {\em Nonasymptotic convergence analysis for
  the unadjusted {L}angevin algorithm}, Ann. Appl. Probab., 27 (2017),
  pp.~1551--1587, \url{https://doi.org/10.1214/16-AAP1238},
  \url{https://doi.org/10.1214/16-AAP1238}.

\bibitem{DuMo2019}
{\sc A.~Durmus and E.~Moulines}, {\em High-dimensional {B}ayesian inference via
  the unadjusted {L}angevin algorithm}, Bernoulli, 25 (2019), pp.~2854--2882,
  \url{https://doi.org/10.3150/18-BEJ1073},
  \url{https://doi.org/10.3150/18-BEJ1073}.

\bibitem{Eb2016}
{\sc A.~Eberle}, {\em Reflection couplings and contraction rates for
  diffusions}, Probab. Theory Related Fields, 166 (2016), pp.~851--886,
  \url{https://doi.org/10.1007/s00440-015-0673-1},
  \url{https://doi.org/10.1007/s00440-015-0673-1}.

\bibitem{EbGuZi2019b}
{\sc A.~Eberle, A.~Guillin, and R.~Zimmer}, {\em Couplings and quantitative
  contraction rates for {L}angevin dynamics}, Ann. Probab., 47 (2019),
  pp.~1982--2010, \url{https://doi.org/10.1214/18-AOP1299},
  \url{https://doi.org/10.1214/18-AOP1299}.

\bibitem{GeCaStRu1995}
{\sc A.~Gelman, J.~B. Carlin, H.~S. Stern, and D.~B. Rubin}, {\em Bayesian data
  analysis}, Chapman and Hall/CRC, 1995.

\bibitem{GoBrMaMo2025}
{\sc N.~Gouraud, P.~L. Bris, A.~Majka, and P.~Monmarch\'{e}}, {\em Hmc and
  underdamped langevin united in the unadjusted convex smooth case}, SIAM/ASA
  Journal on Uncertainty Quantification, 13 (2025), pp.~278--303,
  \url{https://doi.org/10.1137/23M1608963},
  \url{https://doi.org/10.1137/23M1608963}.

\bibitem{Ha1970}
{\sc W.~K. Hastings}, {\em Monte carlo sampling methods using markov chains and
  their applications},  (1970).

\bibitem{KlPl1992}
{\sc P.~E. Kloeden and E.~Platen}, {\em Numerical solution of stochastic
  differential equations}, vol.~23 of Applications of Mathematics (New York),
  Springer-Verlag, Berlin, 1992,
  \url{https://doi.org/10.1007/978-3-662-12616-5},
  \url{https://doi.org/10.1007/978-3-662-12616-5}.

\bibitem{LeMa2013}
{\sc B.~Leimkuhler and C.~Matthews}, {\em Rational construction of stochastic
  numerical methods for molecular sampling}, Appl. Math. Res. Express. AMRX,
  (2013), pp.~34--56, \url{https://doi.org/10.1093/amrx/abs010},
  \url{https://doi.org/10.1093/amrx/abs010}.

\bibitem{LeMaSt2016}
{\sc B.~Leimkuhler, C.~Matthews, and G.~Stoltz}, {\em The computation of
  averages from equilibrium and nonequilibrium {L}angevin molecular dynamics},
  IMA J. Numer. Anal., 36 (2016), pp.~13--79,
  \url{https://doi.org/10.1093/imanum/dru056},
  \url{https://doi.org/10.1093/imanum/dru056}.

\bibitem{LePaWh2024}
{\sc B.~Leimkuhler, D.~Paulin, and P.~A. Whalley}, {\em Contraction rate
  estimates of stochastic gradient kinetic {L}angevin integrators}, ESAIM Math.
  Model. Numer. Anal., 58 (2024), pp.~2255--2286,
  \url{https://doi.org/10.1051/m2an/2024038},
  \url{https://doi.org/10.1051/m2an/2024038}.

\bibitem{LePaWh2024a}
{\sc B.~J. Leimkuhler, D.~Paulin, and P.~A. Whalley}, {\em Contraction and
  convergence rates for discretized kinetic {L}angevin dynamics}, SIAM J.
  Numer. Anal., 62 (2024), pp.~1226--1258,
  \url{https://doi.org/10.1137/23M1556289},
  \url{https://doi.org/10.1137/23M1556289}.

\bibitem{LeSt2016}
{\sc T.~Lelievre and G.~Stoltz}, {\em Partial differential equations and
  stochastic methods in moleculardynamics}, Acta Numerica, 25 (2016),
  pp.~681--880.

\bibitem{MaSm2021}
{\sc O.~Mangoubi and A.~Smith}, {\em Mixing of {H}amiltonian {M}onte {C}arlo on
  strongly log-concave distributions: continuous dynamics}, Ann. Appl. Probab.,
  31 (2021), pp.~2019--2045, \url{https://doi.org/10.1214/20-aap1640},
  \url{https://doi.org/10.1214/20-aap1640}.

\bibitem{McQu2002}
{\sc R.~I. McLachlan and G.~R.~W. Quispel}, {\em Splitting methods}, Acta
  Numerica, 11 (2002), pp.~341--434.

\bibitem{MeRoRoTeTe1953}
{\sc N.~Metropolis, A.~W. Rosenbluth, M.~N. Rosenbluth, A.~H. Teller, and
  E.~Teller}, {\em Equation of state calculations by fast computing machines},
  The journal of chemical physics, 21 (1953), pp.~1087--1092.

\bibitem{Mo2021}
{\sc P.~Monmarch\'e}, {\em High-dimensional {MCMC} with a standard splitting
  scheme for the underdamped {L}angevin diffusion.}, Electron. J. Stat., 15
  (2021), pp.~4117--4166, \url{https://doi.org/10.1214/21-ejs1888},
  \url{https://doi.org/10.1214/21-ejs1888}.

\bibitem{Ne1995}
{\sc R.~M. Neal}, {\em Bayesian learning for neural networks}, vol.~118,
  Springer Science \& Business Media, 2012.

\bibitem{Ne2018}
{\sc Y.~Nesterov}, {\em Lectures on convex optimization}, vol.~137 of Springer
  Optimization and Its Applications, Springer, Cham, second~ed., 2018,
  \url{https://doi.org/10.1007/978-3-319-91578-4},
  \url{https://doi.org/10.1007/978-3-319-91578-4}.

\bibitem{PaWh2024}
{\sc D.~Paulin and P.~A. Whalley}, {\em Correction to" wasserstein distance
  estimates for the distributions of numerical approximations to ergodic
  stochastic differential equations"}, Journal of Machine Learning Research, 25
  (2024), pp.~1--9.

\bibitem{Pa2014}
{\sc G.~A. Pavliotis}, {\em Stochastic processes and applications}, Texts in
  applied mathematics, 60 (2014), pp.~41--43.

\bibitem{SaZy2021}
{\sc J.~M. Sanz-Serna and K.~C. Zygalakis}, {\em Wasserstein distance estimates
  for the distributions of numerical approximations to ergodic stochastic
  differential equations}, J. Mach. Learn. Res., 22 (2021), pp.~Paper No. 242,
  37.

\bibitem{Sc2024}
{\sc K.~Schuh}, {\em Global contractivity for {L}angevin dynamics with
  distribution-dependent forces and uniform in time propagation of chaos}, Ann.
  Inst. Henri Poincar\'e{} Probab. Stat., 60 (2024), pp.~753--789,
  \url{https://doi.org/10.1214/22-aihp1337},
  \url{https://doi.org/10.1214/22-aihp1337}.

\bibitem{ScWh2024}
{\sc K.~Schuh and P.~A. Whalley}, {\em Convergence of kinetic langevin samplers
  for non-convex potentials}, arXiv preprint arXiv:2405.09992,  (2024).

\bibitem{Vi2009}
{\sc C.~Villani}, {\em Hypocoercivity}, vol.~202, American Mathematical
  Society, 2009.

\bibitem{WeTe2011}
{\sc M.~Welling and Y.~W. Teh}, {\em Bayesian learning via stochastic gradient
  {L}angevin dynamics}, in Proceedings of the 28th international conference on
  machine learning (ICML-11), 2011, pp.~681--688.

\bibitem{Za2017}
{\sc A.~A. Zapatero}, {\em Word series for the numerical integration of
  stochastic differential equations}, PhD thesis, Universidad de Valladolid,
  2017.

\end{thebibliography}
\end{document}